\documentclass[acmsmall,nonacm]{acmart}

\acmJournal{PACMPL}
\acmVolume{1}
\acmNumber{CONF} 
\acmArticle{1}
\acmYear{2018}
\acmMonth{1}
\acmDOI{} 
\startPage{1}

\setcopyright{none}

\usepackage{booktabs}   
\usepackage{subcaption} 

\usepackage{amssymb}
\usepackage{amsmath}
\usepackage{amsthm}
\usepackage{tikz}
\usepackage{mathpartir}
\usepackage{dashbox}

\makeatletter
\DeclareRobustCommand{\rvdots}{%
  \vbox{
    \baselineskip4\p@\lineskiplimit\z@
    \kern-\p@
    \hbox{.}\hbox{.}\hbox{.}
  }}
\makeatother

\usetikzlibrary{arrows,shapes,snakes,automata,backgrounds,petri,positioning}

\tikzstyle{place}=[circle,draw,minimum size=5mm]
\tikzstyle{transitionV}=[rectangle,thick,fill=black,minimum width=5mm,
    inner ysep=1pt]
\tikzstyle{transitionH}=[rectangle,thick,fill=black,minimum height=5mm,
    inner xsep=1pt]
\tikzstyle{every picture}+=[node distance=0.75cm,>=stealth',auto]

\newcommand{\cmd}[1]{{\mbox{\bf #1}}}
\newcommand{\action}[1]{\ensuremath{\mathsf{#1}}}
\newcommand{\place}[1]{\ensuremath{\mathrm{#1}}}

\newcommand{\llbrace}{\{\hspace{-2.5pt}[}
\newcommand{\rrbrace}{]\hspace{-2.5pt}\}}

\newcommand\dboxed[1]{\dbox{\ensuremath{#1}}}

\definecolor{Green}{rgb}{0.0,0.5,0.0}
\definecolor{White}{rgb}{1,1,1}
\definecolor{Blue}{rgb}{0.0,0.0,1}
\definecolor{Red}{rgb}{1,0,0}
\definecolor{Purple}{rgb}{0.4,0,0.6}
\definecolor{Orange}{rgb}{1,0.65,0}
\definecolor{Gray}{rgb}{0.5,0.5,0.5}
\newcommand{\annot}[1]{{\color{Blue}#1}}
\newcommand{\codecomment}[1]{{\color{Green}#1}}

\DeclareMathSymbol{\mrq}{\mathord}{operators}{`'}

\begin{document}

\title{Abstract I/O Specification}


\author{Willem Penninckx}
\email{willem@willemp.be}
\author{Amin Timany}
\email{amin.timany@cs.kuleuven.be}
\author{Bart Jacobs}
\orcid{0000-0002-3605-249X}
\email{bart.jacobs@cs.kuleuven.be}
\affiliation{
  \department{Department of Computer Science}
  \institution{KU Leuven}
  \streetaddress{Celestijnenlaan 200A}
  \city{Leuven}
  \postcode{3001}
  \country{Belgium}
}

\begin{abstract}
Penninckx \emph{et al.}~recently proposed an approach for the specification and modular formal verification of the interactive (I/O) behavior of programs, based on an embedding of Petri nets into separation logic. While this approach is scalable and modular in terms of the I/O APIs available to a program, enables composing low-level I/O actions into high-level ones, and enables a convenient verification experience, it does not support high-level I/O actions that involve memory manipulation as well as low-level I/O (such as buffered I/O), or that are in fact ``virtual I/O'' actions that are implemented purely through memory manipulation (such as offered by Java's ByteArrayOutputStream). Furthermore, it does not allow rewriting an I/O specification into an equivalent one.

In this paper, we propose a refined approach that does have these properties. The essential insight is to fix the set of places of the Petri net to be the set of separation logic assertions, thus making available the full power of separation logic for abstractly stating an arbitrary operation's specification in Petri net form, for composing operations into an I/O specification, and for equivalence reasoning on I/O specifications. Our refinement resolves the issue of the justification of the choice of Petri nets over other formalisms such as general state transition systems, in that it ``refines them away'' into the more essential constructs of separating conjunction and \emph{abstract nested triples}. To enable a convenient treatment of input operations, we propose the use of prophecy variables to eliminate their non-determinism.

We illustrate the approach through a number of example programs, including one where subroutines specified and verified using I/O specifications run as threads communicating through shared memory. The theory and examples of the paper have been machine-checked using the Iris library for program verification in the Coq proof assistant.
\end{abstract}



\maketitle

\section{Introduction}

While great progress has been made in recent decades on approaches for modular formal verification of memory safety of imperative programs, as well as functional correctness of data structures and algorithms, even in the presence of pointer manipulation (or other types of aliasing), coarse-grained or fine-grained concurrency, and higher-order programming, the issue of specifying and modularly verifying the actual observable interactive behavior of the program as a whole, such as achieved through I/O APIs including file I/O, network I/O, graphical user interface APIs, etc., has received much less attention; verification of interactive behavior, if done at all, has mostly been performed at the level of abstract models, often using techniques such as model checking, rather than at the level of source code, integrated with the Hoare logic used for verifying memory safety and data structure correctness. This leaves an unverified gap between the abstract models and the source code.

In this paper, we address this issue by proposing an approach for integrating I/O verification into a Hoare-style modular program verification approach where each function of the program is assigned a specification consisting of a precondition and a postcondition, and then each function is verified against its specification, under the assumption that its callees satisfy theirs. This involves in particular addressing the question of what the specifications of the platform's I/O API functions, and the program's main function, should look like, to express the program's behavioral requirements.
Our goal is that the approach should be applicable to an annotation-based verification tool, such as e.g.~the VeriFast tool, for programs written in real languages, such as C or Java, against real platform APIs, such as \verb|stdio.h| or \verb|java.io|.

As far as we know, the only approach that has been proposed so far to address this goal is the one proposed by \citeN{io}, where a program's behavioral requirements are expressed as \emph{Petri nets} embedded into separation logic \cite{seplogic-ohearn}. This approach scales and is modular with respect to the number of I/O APIs available to a program, it allows the program (or program libraries) to define higher-level I/O actions on top of the platform I/O actions (such that a program module can be agnostic as to which actions are primitive and which are composite), and it integrates well into existing separation logic tool support approaches such as symbolic execution with symbolic heaps, yielding a convenient, low-overhead verification experience. However, this approach does not support high-level I/O actions that involve memory manipulation as well as low-level I/O (such as buffered I/O), or that are in fact ``virtual I/O'' actions that are implemented purely through memory manipulation (such as offered by Java's ByteArrayOutputStream). Furthermore, it does not allow rewriting an I/O specification into an equivalent one.

In this paper, we propose a refined approach that does have these properties. The essential insight is to fix the set of places of the Petri net to be the set of separation logic assertions, thus making available the full power of separation logic for abstractly stating an arbitrary operation's specification in Petri net form, for composing operations into an I/O specification, and for equivalence reasoning on I/O specifications. To enable a convenient treatment of non-deterministic input operations, we propose the use of prophecy variables to eliminate the non-determinism.

Our proposed refinement resolves the issue of the justification of the choice of Petri nets over other formalisms, such as general state transition systems, in that while the refined approach can still be seen as applying the Petri nets formalism, it can also be explained straightforwardly without any reference to Petri nets, as simply applying the essential concepts of separating conjunction and an abstract form of \emph{nested Hoare triples} (e.g.~\cite{nested-hoare-triples,iris3}).

We illustrate the approach through a number of example programs, including one where subroutines specified and verified using I/O specifications run as threads communicating through shared memory. The theory and examples of the paper have been machine-checked using the Iris \cite{iris3} library for verification of concurrent programs in the Coq proof assistant.

The rest of this paper is structured as follows. In \S\ref{sec:lang-with-io}, we define the syntax and the semantics of the programming language that we will use to present our approach. In \S\ref{sec:petrinets}, we recall the Petri net-based specification approach \cite{io} that we refine in this work. In \S\ref{sec:buffered-output}, we introduce our refined approach, and we motivate it by means of the example of buffered output. In \S\ref{sec:dealing-with-input}, we illustrate the problem of proving I/O-style specifications for in-memory input operations by means of a chat server example. In \S\ref{sec:prophvars} we introduce prophecy variables to address this problem. In \S\ref{sec:concur} we extend our programming languge from \S\ref{sec:lang-with-io} and our Hoare logic from \S\ref{sec:petrinets} to support concurrency, which we then use in \S\ref{sec:channels-proof} to verify an implementation of the channels construct used in the chat server against I/O-style specifications. We end the paper with a discussion of related work (\S\ref{sec:related-work}) and a conclusion (\S\ref{sec:conclusion}).

\section{A Programming Language with I/O}\label{sec:lang-with-io}

\subsection{The Programming Language}\label{sec:labeled}

We present the basic idea of our approach in the context of a simple ML-like programming language with support for I/O. Its grammar is as follows:

$$\begin{array}{r @{\;} l}
& x \in \mathit{Vars}, t \in \mathit{IOTags}\\
e \in \mathit{Exprs} ::= & ()\ |\ \mathbf{inl}(e)\ |\ \mathbf{inr}(e)\ |\ (e, e)\ |\ \lambda x.\,e\ |\ x\\
& |\ \mathbf{cases}(e, e, e)\ |\ \mathbf{fst}(e)\ |\ \mathbf{snd}(e)\ |\ e(e)\ |\ \mathbf{assert}(e)\\
& |\ \mathbf{ref}(e)\ |\ !e\ |\ e \leftarrow e\ |\ t(e)
\end{array}$$

We assume a set $\mathit{Vars}$ of program variables and $\mathit{IOTags}$ of \emph{primitive I/O tags}.

We define $\mathbf{let}\ x := e\ \mathbf{in}\ e' = (\lambda x.\;e')(e)$ and $e; e' = \mathbf{let}\ x := e\ \mathbf{in}\ e'$ where $x$ does not appear in $e'$.
We define $\mathbf{true} = \mathbf{inl}(())$ and $\mathbf{false} = \mathbf{inr}(())$, and $\mathbf{if}\ e\ \mathbf{then}\ e_1\ \mathbf{else}\ e_2 = \mathbf{cases}(e, \lambda x.\,e_1, \lambda x.\,e_2)$ where $x$ does not appear in $e_1$ or $e_2$. Furthermore, we define $\mathbf{nil} = \mathbf{inl}(())$ and $\mathbf{cons}(e, e') = \mathbf{inr}((e, e'))$. We encode characters as tuples of booleans and strings as lists of characters.

To define the language's semantics, we define the values and the evaluation contexts as follows:
$$\begin{array}{r @{\;} l}
& \ell \in \mathit{Locs}\\
v \in \mathit{Vals} ::= & ()\ |\ \mathbf{inl}(v)\ |\ \mathbf{inr}(v)\ |\ (v, v)\ |\ \lambda x.\,e\ |\ \ell\\
K \in \mathit{Ctxts} ::= & \bullet\ |\ \mathbf{inl}(K)\ |\ \mathbf{inr}(K)\ |\ (K, e)\ |\ (v, K)\\
& |\ \mathbf{cases}(K, e, e)\ |\ \mathbf{fst}(K)\ |\ \mathbf{snd}(K)\ |\ K(e)\ |\ v(K)\ |\ \mathbf{assert}(K)\\
& |\ \mathbf{ref}(K)\ |\ !K\ |\ K \leftarrow e\ |\ v \leftarrow K\ |\ t(K)
\end{array}$$

We assume an infinite set $\mathit{Locs}$ of heap locations.

We define the I/O actions $\alpha \in \mathit{IOActions} ::= t(v, v)$; in $t(v, v')$, we call $v$ the \emph{argument} and $v'$ the \emph{result}. The traces $\tau \in \mathit{Traces} = \mathit{IOActions}^*$ are the lists of I/O actions. We use $\epsilon$ to denote the empty list and ${-} \cdot {-}$ to denote list concatenation.

We define the heaps $h \in \mathit{Heaps} = \mathit{Locs} \rightharpoonup_\mathrm{fin} \mathit{Vals}$ as the finite partial functions from heap locations to values.

We define the configurations $\gamma \in \mathit{Configs} = \mathit{Heaps} \times \mathit{Exprs}$. We define a labeled head reduction relation $\gamma \stackrel{\tau}{\hookrightarrow} \gamma'$, a labeled small-step relation $\gamma \stackrel{\tau}{\rightarrow} \gamma'$, and a labeled reachability relation $\stackrel{\tau}{\rightarrow^*}$ in Figure~\ref{fig:steprules}.

\begin{figure}
\begin{mathpar}
h, \mathbf{cases}(\mathbf{inl}(v), \lambda x.\,e, \_) \stackrel{\epsilon}{\hookrightarrow} h, e[v/x]
\and
h, \mathbf{cases}(\mathbf{inr}(v), \_, \lambda x.\,e) \stackrel{\epsilon}{\hookrightarrow} h, e[v/x]
\and
h, \mathbf{fst}((v, v')) \stackrel{\epsilon}{\hookrightarrow} h, v
\and
h, \mathbf{snd}((v, v')) \stackrel{\epsilon}{\hookrightarrow} h, v'
\and
h, (\lambda x.\,e)(v) \stackrel{\epsilon}{\hookrightarrow} h, e[v/x]
\and
h, \mathbf{assert}(\mathbf{true}) \stackrel{\epsilon}{\hookrightarrow} h, ()
\and
\inferrule{
\ell \notin \mathrm{dom}(h)
}{
h, \mathbf{ref}(v) \stackrel{\epsilon}{\hookrightarrow} h[\ell := v], \ell
}
\and
\inferrule{
\ell \in \mathrm{dom}(h)
}{
h, !\ell \stackrel{\epsilon}{\hookrightarrow} h, h(\ell)
}
\and
\inferrule{
\ell \in \mathrm{dom}(h)
}{
h, \ell \leftarrow v \stackrel{\epsilon}{\hookrightarrow} h[\ell := v], ()
}
\and
h, t(v) \stackrel{t(v, v')}{\hookrightarrow} h, v'
\and
\inferrule{
h, e \stackrel{\tau}{\hookrightarrow} h, e'
}{
h, K[e/\bullet] \stackrel{\tau}{\rightarrow} h, K[e'/\bullet]
}
\and
\gamma \stackrel{\epsilon}{\rightarrow^*} \gamma
\and
\inferrule{
\gamma \stackrel{\tau}{\rightarrow} \gamma'\\
\gamma' \stackrel{\tau'}{\rightarrow^*} \gamma''
}{
\gamma \stackrel{\tau\cdot\tau'}{\rightarrow^*} \gamma''
}
\end{mathpar}
\caption{The labeled head reduction relation $\hookrightarrow$, the labeled small-step relation $\rightarrow$, and the labeled reachability relation $\rightarrow^*$}\label{fig:steprules}
\end{figure}

We say a configuration is finished if its expression is a value: $\mathsf{finished}(h, e) \Leftrightarrow e \in \mathit{Vals}$, and that it has failed if it is not finished and not reducible: $\mathsf{failed}(\gamma) \Leftrightarrow \lnot \mathsf{finished}(\gamma) \land \not\exists \tau, \gamma'.\;\gamma \stackrel{\tau}{\rightarrow} \gamma'$.

\subsection{I/O Specifications}\label{sec:iospecs}

A foundational way of specifying the desired I/O behavior of a program is in the form of a prefix-closed\footnote{A set $T$ is prefix-closed if $\tau\cdot\tau' \in T$ implies $\tau \in T$.} set $T$ of traces. We say a configuration $\gamma$ satisfies such a specification, denoted $\gamma \vDash T$, if for any configuration $\gamma'$ reachable from $\gamma$ via a trace $\tau \in T$ (implying that both the program and the environment behave according to $\tau$), $\gamma'$ has not failed and furthermore for any I/O action $t(v, v')$ that $\gamma'$ can perform, the trace $\tau\cdot t(v, v'')$ is in $T$, for some $v''$:
$$\gamma \vDash T \Leftrightarrow \forall \tau \in T, \gamma'.\;\gamma \stackrel{\tau}{\rightarrow^*} \gamma' \Rightarrow \lnot\mathsf{failed}(\gamma') \land \forall t, v, v', \gamma''.\;\gamma' \stackrel{t(v,v')}{\rightarrow} \gamma'' \Rightarrow \exists v''.\;\tau\cdot t(v, v'') \in T$$

For example ($\tau \preceq \tau'$ denotes that $\tau$ is a prefix of $\tau'$: $\exists \tau''.\;\tau' = \tau\cdot\tau''$):
$$\mathsf{putbool}(\lnot \mathsf{getbool}(())) \vDash \{\tau\ |\ \tau \preceq \mathsf{getbool}((), b)\cdot\mathsf{putbool}(\lnot b, ()), b \in \{\mathbf{true}, \mathbf{false}\}\}$$
where in the program $\lnot e = \mathbf{if}\ e\ \mathbf{then}\ \mathbf{false}\ \mathbf{else}\ \mathbf{true}$.
This specification constrains both the program and the environment: it specifies that $\mathsf{getbool}$ shall return only booleans, and that the program's first action, if any, shall be to get a boolean, and its second action, if any, shall be to put its negation, and that it shall not perform any further actions. The program is allowed to get stuck (and it generally does) if $\mathsf{getbool}$ returns something other than a boolean.

In this paper, we focus on safety properties only; we do not consider verifying termination or liveness properties. Still, we may wish to express that a program $e$ satisfies specification $T$ and that furthermore, if it terminates, it shall have performed a trace from set $T'_v$, where $v$ is the program's result. We can encode this by extending the set of I/O tags with an $\mathsf{exit}$ tag and specifying that $\mathsf{exit}(e) \vDash T \cup \{\tau\ |\ \tau \preceq \tau'\cdot\mathsf{exit}(v), \tau' \in T'_v\}$.

\subsection{An Unlabeled Semantics}\label{sec:monitoring}

Most modular program verification approaches proposed in the literature assume an unlabeled operational semantics and simply verify that the program does not reach a failed configuration. Fortunately, we can encode satisfaction of an I/O specification into a statement of this form by using a \emph{monitoring} semantics, defined by an unlabeled small-step relation over \emph{instrumented configurations} which include an I/O specification in the form of a prefix-closed set of traces.

For the example programming language, in the monitoring semantics, the step rule for I/O expressions is as follows:
$$\inferrule{
t(v,v') \in T
}{
T, h, t(v) \hookrightarrow \{\tau\ |\ t(v, v')\cdot\tau \in T\}, h, v'
}$$
The other step rules do not affect, and are not affected by, the I/O specification.

\begin{lemma}
If $\tau \in T$ and $h, e \stackrel{\tau}{\rightarrow^*} h', e'$ then $(T, h, e) \rightarrow^* (\{\tau'\ |\ \tau\cdot\tau' \in T\}, h', e')$.
\end{lemma}

We say a configuration $\gamma$ is safe, denoted $\mathsf{safe}(\gamma)$, if no failed configuration is reachable from it.

\begin{theorem}
If $\mathsf{safe}((T, h, e))$ then $h, e \vDash T$.
\end{theorem}

\section{Recap of the Petri net approach}\label{sec:petrinets}

In this section, we recall the I/O specification approach presented by \citeN{io}.\footnote{Our presentation differs in unimportant ways from that of \citeN{io}.} In subsequent sections, we propose a number of refinements to this approach, to achieve more abstract I/O specifications.

\subsection{Petri nets for I/O specification}

A Petri net is defined by a set of \emph{places} $\mathcal{P}$, ranged over by $p$ and $q$, and a set $N$ of \emph{transitions}. A \emph{marking} $V \in \mathit{Markings} = \mathcal{P} \rightarrow \mathbb{N}$ of a Petri net maps each place to the number of \emph{tokens} present at that place. Given a marking, a transition can \emph{fire} if there is a token at each of its \emph{pre-places}. Firing the transition removes one token from each of the transition's pre-places and adds one token to each of its \emph{post-places}.

We use the following notation for markings: $V \uplus V' = \lambda p.\,V(p) + V'(p)$; $\mathbf{0} = \lambda p.\,0$; $\llbrace p\rrbrace = \mathbf{0}[p:=1]$; $\llbrace p, q\rrbrace = \llbrace p\rrbrace \uplus \llbrace q\rrbrace$.

Petri nets can be used to denote I/O specifications by labeling some transitions with I/O actions. In particular, we will use Petri nets whose transitions $\nu$ are of the following form:
$$\nu \in \mathcal{N} ::= t(p, v, v, q)\ |\ \mathbf{split}(p, q, q')\ |\ \mathbf{join}(p, p', q)\ |\ \mathbf{noop}(p, q)$$
where $p, p' \in \mathcal{P}$ are the pre-places and $q, q' \in \mathcal{P}$ are the post-places.

A Petri net, given by its set $N \subseteq \mathcal{N}$ of transitions, defines a labeled step relation $\rightarrow$ and a corresponding labeled reachability relation $\rightarrow^*$ on markings:
\begin{mathpar}
\inferrule{
t(p, v, v', q) \in N
}{
V \uplus \llbrace p\rrbrace \stackrel{t(v, v')}{\rightarrow} V \uplus \llbrace q\rrbrace
}
\and
\inferrule{
\mathbf{split}(p, q, q') \in N
}{
V \uplus \llbrace p\rrbrace \stackrel{\epsilon}{\rightarrow} V \uplus \llbrace q, q'\rrbrace
}
\and
\inferrule{
\mathbf{join}(p, p', q) \in N
}{
V \uplus \llbrace p, p'\rrbrace \stackrel{\epsilon}{\rightarrow} V \uplus \llbrace q\rrbrace
}
\and
\inferrule{
\mathbf{noop}(p, q) \in N
}{
V \uplus \llbrace p\rrbrace \stackrel{\epsilon}{\rightarrow} V \uplus \llbrace q\rrbrace
}
\end{mathpar}

We define $\mathsf{Traces}_N(V) = \{\tau\ |\ \exists V'.\;V \stackrel{\tau}{\rightarrow^*} V'\}$. Notice that this set is always prefix-closed.

\subsection{A Separation Logic for I/O Verification}\label{sec:hoare-logic}

We can verify that a program satisfies the I/O specification implied by a marking of a Petri net by means of a Hoare logic (more specifically: a separation logic) whose assertions describe a heap and a marking: $\mathit{Asns} = \frak{P}(\mathit{Heaps} \times \mathit{Markings})$.\footnote{$\frak{P}(X)$ denotes the powerset of $X$.}

We define $\ell \mapsto v = \{(h, V)\ |\ \ell \in \mathrm{dom}(h) \land h(\ell) = v\}$ and $\mathbf{token}(p) = \{(h, V)\ |\ V(p) > 0\}$.
We define $P * P' = \{(h'', V \uplus V')\ |\ (h, V) \in P \land (h', V') \in P' \land h'' = h \uplus h'\}$ where $h'' = h \uplus h'$ means $\mathrm{dom}(h) \cap \mathrm{dom}(h') = \emptyset \land \mathrm{Graph}(h'') = \mathrm{Graph}(h) \cup \mathrm{Graph}(h')$.

We define $T, h \vDash P \Leftrightarrow \exists V.\;\mathsf{Traces}_N(V) \subseteq T \land (h, V) \in P$.

We define the meaning of correctness judgments:
$$\{P\}\ e\ \{Q\} \Leftrightarrow \forall T \in \mathit{ResDet}, h.\;T, h \vDash P \Rightarrow \mathsf{safe}(T, h, e, Q)$$
where
$$\mathit{ResDet} = \{T\ |\ \forall \tau, t, v, v_1, v_2.\; \tau\cdot t(v, v_1) \in T \land \tau\cdot t(v, v_2) \in T \Rightarrow v_1 = v_2\}$$
and
$$\mathsf{safe}(T, h, e, Q) \Leftrightarrow \mathsf{safe}((T, h, e)) \land \forall T', h', v.\;(T, h, e) \rightarrow^*(T', h', v) \Rightarrow T', h' \vDash Q(v)$$
and postconditions $Q : \mathit{Vals} \rightarrow \mathit{Asns}$ are functions from values to assertions. We lift operations on assertions pointwise to operations on postconditions. Also, we usually write postconditions using the notation $P(\mathsf{res})$ instead of $\lambda v.\;P(v)$, where $\mathsf{res}$ stands for \emph{result}.

The logic supports only \emph{result-deterministic} I/O specifications, i.e.~ones that do not underspecify the results of I/O actions. However, this is not a significant restriction, since any I/O specification T can be written as the union of a set $\mathcal{T}$ of result-deterministic I/O specifications, and we have the property $(\forall T \in \mathcal{T}.\;\gamma \vDash T) \Rightarrow \gamma \vDash \bigcup \mathcal{T}$. For each $T \in \mathcal{T}$, $\gamma \vDash T$ can be verified using the Hoare logic.

We say an assertion $P$ \emph{precedes} an assertion $P'$, denoted $P \sqsubseteq P'$, if $\forall (h, V) \in P.\;\exists V'.\;
\mathsf{Traces}_N(V') \subseteq \mathsf{Traces}_N(V) \land (h, V') \in P'$.

From these definitions, we can derive the proof rules shown in Figure~\ref{fig:hoare-logic}.

\begin{figure}
\begin{mathpar}
\{Q(v)\}\ v\ \{Q\}
\and
\inferrule{
\{P\}\ e[v/x]\ \{Q\}
}{
\{P\}\ \mathbf{cases}(\mathbf{inl}(v), \lambda x.\;e, \_)\ \{Q\}
}
\and
\inferrule{
\{P\}\ e[v/x]\ \{Q\}
}{
\{P\}\ \mathbf{cases}(\mathbf{inr}(v), \_, \lambda x.\;e)\ \{Q\}
}
\and
\inferrule{
\{P\}\ v_1\ \{Q\}
}{
\{P\}\ \mathbf{fst}((v_1, v_2))\ \{Q\}
}
\and
\inferrule{
\{P\}\ v_2\ \{Q\}
}{
\{P\}\ \mathbf{snd}((v_1, v_2))\ \{Q\}
}
\and
\inferrule{
\{P\}\ e[v/x]\ \{Q\}
}{
\{P\}\ (\lambda x.\,e)(v)\ \{Q\}
}
\and
\{\mathsf{True}\}\ \mathbf{ref}(v)\ \{\mathsf{res} \mapsto v\}
\and
\{\ell \mapsto v\}\ !\ell\ \{\ell \mapsto v \land \mathsf{res} = v\}
\and
\{\ell \mapsto \_\}\ \ell \leftarrow v\ \{\ell \mapsto v\}
\and
\{\mathbf{token}(p) \land t(p, v, v', q) \in N\}\ t(v)\ \{\mathbf{token}(q) \land \mathsf{res} = v'\}
\and
\inferrule{
\{P\}\ e\ \{Q\}\\
\forall v.\;\{Q(v)\}\ K[v/\bullet]\ \{Q'\}
}{
\{P\}\ K[e/\bullet]\ \{Q'\}
}
\and
\inferrule{
P \sqsubseteq P'\\
\{P'\}\ e\ \{Q\}\\
Q \sqsubseteq Q'
}{
\{P\}\ e\ \{Q'\}
}
\and
\inferrule{
\{P\}\ e\ \{Q\}
}{
\{P * R\}\ e\ \{Q * R\}
}
\and
\inferrule{
\forall i \in I.\;\{P_i\}\ e\ \{Q\}
}{
\{\bigcup_i P_i\}\ e\ \{Q\}
}
\and
\inferrule{
P \subseteq P'
}{
P \sqsubseteq P'
}
\and
\inferrule{
P \sqsubseteq P'
}{
P * R \sqsubseteq P' * R
}
\and
\inferrule{
\mathbf{split}(p, q, q') \in N
}{
\mathbf{token}(p) \sqsubseteq \mathbf{token}(q) * \mathbf{token}(q')
}
\and
\inferrule{
\mathbf{join}(p, p', q) \in N
}{
\mathbf{token}(p) * \mathbf{token}(p') \sqsubseteq \mathbf{token}(q)
}
\and
\inferrule{
\mathbf{noop}(p, q) \in N
}{
\mathbf{token}(p) \sqsubseteq \mathbf{token}(q)
}
\end{mathpar}
\caption{Proof rules of the separation logic for I/O verification}\label{fig:hoare-logic}
\end{figure}

\subsection{Examples}\label{sec:examples}

The following diagram denotes a Petri net with places $\mathcal{P} = \{\place{p_1}, \place{p_2}, \place{p_3}\}$, marking $V = \llbrace \place{p_1}\rrbrace$, and transitions $N = \{\action{putchar}(\place{p_1}, \texttt{'h'}, (), \place{p_2}), \action{putchar}(\place{p_2}, \texttt{'i'}, (), \place{p_3})\}$, where $\action{putchar} \in \mathit{IOTags}$ is an I/O tag:
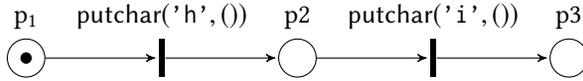
\begin{center}
\begin{tikzpicture}[node distance=0.4cm and 1.5cm]
  \node [place,tokens=1,label=above:$\place{p_1}$] (p1) {};
  \node [transitionH] (pc1) [right=of p1, label=above:{$\action{putchar}(\texttt{'h'}, ())$}] {}
    edge [pre] (p1);
  \node [place] (p2) [right=of pc1, label=above:$\place{p2}$] {}
    edge [pre] (pc1);
  \node [transitionH] (pc2) [right=of p2, label=above:{$\action{putchar}(\texttt{'i'}, ())$}] {}
    edge [pre] (p2);
  \node [place] (p3) [right=of pc2, label=above:$\place{p3}$] {}
    edge [pre] (pc2);
\end{tikzpicture}
\end{center}
When used as an I/O specification, it allows the program to perform the I/O actions $\action{putchar}(\texttt{'h'}, ())$ and $\action{putchar}(\texttt{'i'}, ())$, once, in that order. Indeed, given the marking shown (one token in place $\place{p1}$ and zero tokens in places $\place{p2}$ and $\place{p3}$), the transition labeled $\action{putchar}(\texttt{'h'}, ())$ can \emph{fire} (because all of its \emph{pre-places} have a token), which removes one token from each of the transition's pre-places and adds one to each of its \emph{post-places}. (No other transition can fire initially.) If it does, in the resulting marking, only the other transition can fire, etc.

Per the Hoare logic presented above, for any places $p, q \in \mathcal{P}$ we have the following Hoare triple for the I/O expression $\mathsf{putchar}(c)$, where $c$ is a character:
$$\annot{\{\mathbf{token}(p) \land \mathsf{putchar\_}(p, c, (), q)\}}\ \mathsf{putchar}(c)\label{def:putchar}\ \annot{\{\mathbf{token}(q)\}}$$
where we use the shorthand $t\_(p, v, v', q)$ for $t(p, v, v', q) \in N$. In the remainder, we will abbreviate the action $\mathsf{putchar}(v, ())$ to $\mathsf{putchar}(v)$ and the transition $\mathsf{putchar}(p, v, (), q)$ to $\mathsf{putchar}(p, v, q)$.

Assuming the Petri net above, the following Hoare triple expresses that a function $\mathsf{main}$ satisfies the I/O specification denoted by the Petri net:\footnote{We use notation $f()$ to abbreviate function application $f(())$.}
\begin{equation}\label{eqn:hi-spec-concrete}
\annot{\{\mathbf{token}(\place{p_1})\}}\ \mathsf{main}()\ \annot{\{\mathbf{token}(\place{p_3})\}}
\end{equation}
However, when specifying functions, it is preferable that the Hoare triple itself express any necessary assumptions about the Petri net, like so:
$$\forall \mathcal{P}, N, p_1, p_2, p_3.\;\annot{\left\{\begin{array}{@{} l @{}}
\mathbf{token}(p_1) \land \mathsf{putchar\_}(p_1, \texttt{'h'}, p_2)\\
{} \land \mathsf{putchar\_}(p_2, \texttt{'i'}, p_3)
\end{array}\right\}}\ \mathsf{main}()\ \annot{\{\mathbf{token}(p_3)\}}$$
This specification universally quantifies over the set of places $\mathcal{P}$, the set of transitions $N$, and the places $p_1$, $p_2$, and $p_3$. It is easy to see that this specification is indeed equivalent to specification \ref{eqn:hi-spec-concrete} above, in terms of the I/O traces $\mathsf{main}$ is allowed to produce. In the remainder, we will always implicitly universally quantify over the set of places, the set of transitions, and any free metavariables (including ones ranging over places) of a specification.

Consider now the following implementation of function $\mathsf{main}$:\footnote{We use notation $\mathbf{function}\ f()\ \{\ e\ \}$ to mean $f = \lambda x.\;e$, where $x$ does not appear in $e$. Similarly, we use $\mathbf{function}\ f(x)\ \{\ e\ \}$ to mean $f = \lambda x.\;e$ and $\mathbf{function}\ f(x, y)\ \{\ e\ \}$ to mean $f = \lambda \mathit{args}.\;e[\mathbf{fst}(\mathit{args})/x, \mathbf{snd}(\mathit{args})/y]$.}
$$\mathbf{function}\ \mathsf{main}()\ \{\ \mathsf{putchar}(\texttt{'h'}); \mathsf{putchar}(\texttt{'i'})\ \}$$
We can verify that this function satisfies its specification using the Hoare rules from Figure~\ref{fig:hoare-logic}. Such a proof is commonly summarized as a \emph{Hoare proof outline} that mentions the most salient intermediate assertions, like so:
$$\begin{array}{l}
\mathbf{function}\ \mathsf{main}()\ \{\\
\quad \annot{\{\mathbf{token}(p_1) \land \mathsf{putchar\_}(p_1, \texttt{'h'}, p_2) \land \mathsf{putchar\_}(p_2, \texttt{'i'}, p_3)\}}\\
\quad \mathsf{putchar}(\texttt{'h'});\\
\quad \annot{\{\mathbf{token}(p_2) \land \mathsf{putchar\_}(p_2, \texttt{'i'}, p_3)\}}\\
\quad \mathsf{putchar}(\texttt{'i'})\\
\quad \annot{\{\mathbf{token}(p_3)\}}\\
\}
\end{array}$$

\subsubsection{Underspecification}

One can easily express specifications that allow multiple behaviors, by specifying a Petri net where there are multiple paths between the start and destination places. For example:
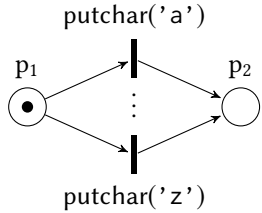
\begin{center}
\begin{tikzpicture}[node distance=0.1cm and 1cm]
  \node [place,tokens=1,label=above:$\mathrm{p_1}$] (t1) {};
  \node [right=of t1] (ts) {\rvdots};
  \node [transitionH] (ta) [above=of ts, label=above:$\action{putchar}(\texttt{'a'})$] {}
    edge [pre] (t1);
  \node [transitionH] (tz) [below=of ts, label=below:$\action{putchar}(\texttt{'z'})$] {}
    edge [pre] (t1);
  \node [place,right=of ts,label=above:$\mathrm{p_2}$] (t2) {}
    edge [pre] (ta)
    edge [pre] (tz);
\end{tikzpicture}%
\end{center}
The corresponding Hoare triple, along with one example implementation that satisfies it, is as follows:
$$\begin{array}{l}
\mathbf{function}\ \mathsf{put\_some\_char}()\ \{\\
\quad \annot{\{\mathbf{token}(p_1) \land \forall c \in \{\mathtt{\mrq a\mrq}, \dots, \mathtt{\mrq z\mrq}\}.\; \mathsf{putchar\_}(p_1, c, p_2)\}}\\
\quad \mathsf{putchar}(\texttt{'a'})\\
\quad \annot{\{\mathbf{token}(p_2)\}}\\
\}
\end{array}$$

\subsubsection{Compositionality}\label{sec:composite}

The approach allows one to define \emph{composite I/O actions} on top of primitive ones, such that client code need not be aware of whether a given action is primitive or not. For example, we can define a composite I/O action $\mathsf{putchars}$ as follows:\footnote{We use $\mathbf{match}$ notation to abbreviate the corresponding combinations of $\mathbf{cases}$, $\mathbf{fst}$, and $\mathbf{snd}$.}
$$\begin{array}{l}
\mathsf{putchars\_}(p_1, \epsilon, p_2) = \mathbf{noop\_}(p_1, p_2)\\
\mathsf{putchars\_}(p_1, c\cdot\overline{c}, p_3) =\\
\quad \exists p_2.\;\mathsf{putchar\_}(p_1, c, p_2) \land \mathsf{putchars\_}(p_2, \overline{c}, p_3)\\
\\
\mathbf{function}\ \mathsf{putchars}(\overline{c})\ \{\\
\quad \annot{\{\mathbf{token}(p) \land \mathsf{putchars\_}(p, \overline{c}, q)\}}\\
\quad \mathbf{match}\ \overline{c}\ \mathbf{with}\\
\quad |\ \epsilon \Rightarrow ()\\
\quad |\ c\cdot\overline{c}' \Rightarrow \mathsf{putchar}(c); \mathsf{putchars}(\overline{c}')\\
\quad \annot{\{\mathbf{token}(q)\}}\\
\}
\end{array}$$
Notice that if $\overline{c} = \epsilon$, then $\mathbf{token}(p) \sqsubseteq \mathbf{token}(q)$. Notice also that verification of client code can proceed as if $\mathsf{putchars}$ were a primitive I/O tag instead of a function, and $\mathsf{putchars}(\overline{c})$ were a primitive I/O expression instead of a function application, and $\mathsf{putchars\_}$ referred directly to a transition of the Petri net instead of being a predicate defined on top of it.

\subsubsection{Input}

The approach allows one to express input-dependent output requirements, as well as assumptions about the input that will be received. For example, assume $\mathsf{getchar} \in \mathit{IOTags}$. Then for all characters $c$ we have the following Hoare triple:\footnote{Notation $\mathsf{getchar}(p, c, q)$ abbreviates transition $\mathsf{getchar}(p, (), c, q)$.}
$$\annot{\{\mathbf{token}(p) \land \mathsf{getchar\_}(p, c, q)\}}\label{def:getchar}\ \mathsf{getchar}()\ \annot{\{\mathsf{res} = c \land \mathbf{token}(q)\}}$$

The following specification expresses that function $\mathsf{toUpper}$ shall perform a $\mathsf{getchar}$ action, that this action's result shall be a lowercase letter (a constraint on the environment), and that the program shall subsequently output the uppercase version of that letter:
$$\begin{array}{l}
\mathbf{function}\ \mathsf{toUpper}()\ \{\\
\quad \annot{\forall \mathcal{P}, N, p_1, c, p_2, p_3.}\\
\quad \annot{\left\{\begin{array}{@{} l @{}}
\mathbf{token}(p_1) \land \mathsf{getchar\_}(p_1, c, p_2) \land \mathsf{putchar\_}(p_2, c - \texttt{'a'} + \texttt{'A'}, p_3)\\
{} \land \texttt{'a'} \le c \le \texttt{'z'}
\end{array}\right\}}\\
\quad \mathbf{let}\ \mathsf{ch} := \mathsf{getchar}()\ \mathbf{in}\ \mathsf{putchar}(\mathsf{ch} - \texttt{'a'} + \texttt{'A'})\\
\quad \annot{\{\mathbf{token}(p_3)\}}\\
\}
\end{array}$$
(For clarity, we here show the universal quantifications that we will usually leave implicit.) Any particular result-deterministic Petri net that satisfies the precondition has only one $\mathsf{getchar}$ transition starting in $p_1$. However, since the specification is universally quantified over all such Petri nets, it implies that the program properly handles all 26 letters.\footnote{The restriction to result-deterministic Petri nets is implied by the semantics of Hoare triples given in \S\ref{sec:hoare-logic}.}

\subsubsection{Specification-level concurrency}

Suppose we want the program to read two characters and print them back to us. We do not want to force the program to print the first character before it reads the second character. We even want to allow the program to read the second character while, concurrently, it is printing the first character.\footnote{We treat concurrency in the program formally in \S\ref{sec:concur}.} A Petri net that expresses this specification is as follows:

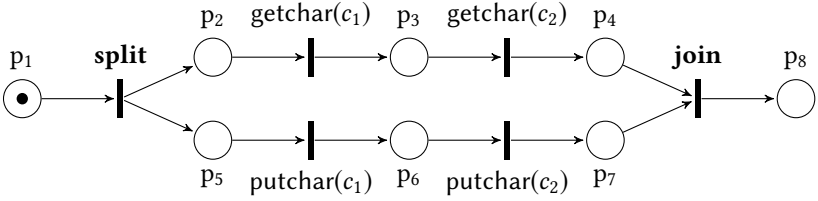
\begin{center}
\begin{tikzpicture}[node distance=0.1cm and 1cm]
  \node [place,tokens=1,label=above:$\mathrm{p_1}$] (t1) {};
  \node [transitionH] (asplit) [right=of t1, label=above:$\cmd{split}$] {}
    edge [pre] (t1);
  \node [place] (ta1) [above right=of asplit, label=above:$\mathrm{p_2}$] {}
    edge [pre] (asplit);
  \node [transitionH] (abeep) [right=of ta1, label=above:$\action{getchar}(c_1)$] {}
    edge [pre] (ta1);
  \node [place] (ta2) [right=of abeep, label=above:$\mathrm{p_3}$] {}
    edge [pre] (abeep);
  \node [transitionH] (aa) [right=of ta2, label=above:{$\action{getchar}(c_2)$}] {}
    edge [pre] (ta2);
  \node [place] (ta3) [right=of aa, label=above:$\mathrm{p_4}$] {}
    edge [pre] (aa);
    
  \node [place] (tb1) [below right=of asplit, label=below:$\mathrm{p_5}$] {}
    edge [pre] (asplit);
  \node [transitionH] (ab) [right=of tb1, label=below:{$\action{putchar}(c_1)$}] {}
    edge [pre] (tb1);
  \node [place] (tb2) [right=of ab, label=below:$\mathrm{p_6}$] {}
    edge [pre] (ab);
  \node [transitionH] (ac) [right=of tb2, label=below:{$\action{putchar}(c_2)$}] {}
    edge [pre] (tb2);
  \node [place] (tb3) [right=of ac, label=below:$\mathrm{p_7}$] {}
    edge [pre] (ac);
  
  \node [transitionH] (ajoin) [below right=of ta3, label=above:$\cmd{join}$] {}
    edge [pre] (ta3)
    edge [pre] (tb3);
  \node [place] (t2) [right=of ajoin, label=above:$\mathrm{p_8}$] {}
    edge [pre] (ajoin);
\end{tikzpicture}%
\end{center}
When the $\mathbf{split}$ transition fires, it consumes the token at $\place{p_1}$ and produces two tokens: one at $\place{p_2}$ and another one at $\place{p_5}$. (Notice that this specification even allows the program to print the first character (and the second character!) while reading the first character, but of course, that is not physically possible.)

The corresponding separation logic specification, with a matching proof outline, is as follows:
$$\begin{array}{l}
\mathbf{function}\ \mathsf{cat2}()\ \{\\
\quad \annot{\left\{\begin{array}{l}
\mathbf{token}(p_1) \land \mathbf{split\_}(p_1, p_2, p_5) \land \mathsf{getchar\_}(p_2, c_1, p_3) \land \mathsf{getchar\_}(p_3, c_2, p_4)\\
{} \land \mathsf{putchar\_}(p_5, c_1, p_6) \land \mathsf{putchar\_}(p_6, c_2, p_7) \land \mathbf{join\_}(p_4, p_7, p_8)
\end{array}\right\}}\\
\quad \codecomment{//\ \mathbf{token}(p_1) \sqsubseteq \mathbf{token}(p_2) * \mathbf{token}(p_5)}\\
\quad \annot{\{\mathbf{token}(p_2) * \mathbf{token}(p_5)\}}\\
\quad \mathbf{let}\ \mathsf{ch1} := \mathsf{getchar}()\ \mathbf{in}\ \mathbf{let}\ \mathsf{ch2} := \mathsf{getchar}()\ \mathbf{in}\\
\quad \annot{\{\mathbf{token}(p_4) * \mathbf{token}(p_5) \land \mathsf{ch1} = c_1 \land \mathsf{ch2} = c_2\}}\\
\quad \mathsf{putchar}(\mathsf{ch1}); \mathsf{putchar}(\mathsf{ch2})\\
\quad \codecomment{//\ \mathbf{token}(p_4) * \mathbf{token}(p_7) \sqsubseteq \mathbf{token}(p_8)}\\
\quad \annot{\{\mathbf{token}(p_8)\}}\\
\}
\end{array}$$

A somewhat more realistic version of $\mathsf{cat}$ is one that reads characters forever and prints them back at its leasure:
$$\annot{\{\mathbf{token}(p_1) * \mathbf{token}(p_2) \land \mathsf{getchars\_}(p_1, \vec{c}) \land \mathsf{putchars\_}(p_2, \vec{c})\}}\ \mathsf{cat}()\ \annot{\{\mathsf{False}\}}$$
where
$$\begin{array}{l}
\mathsf{getchars\_}(p_1, c\cdot\vec{c}) = \exists p_2.\; \mathsf{getchar\_}(p_1, c, p_2) \land \mathsf{getchars\_}(p_2, \vec{c})\\
\mathsf{putchars\_}(p_1, c\cdot\vec{c}) = \exists p_2.\; \mathsf{putchar\_}(p_1, c, p_2) \land \mathsf{putchars\_}(p_2, \vec{c})
\end{array}$$
Here, we intend the weakest solution of these equations; they describe an infinite Petri net. $\vec{c}$ denotes an infinite sequence of characters.

\section{Assertions as places}\label{sec:buffered-output}

\subsection{Motivating example: buffered output}

On Unix-like systems, $\mathsf{putchar}$ is a C run-time library function that is implemented in terms of the $\mathsf{write}$ system call, which writes a sequence of characters:
$$\begin{array}{l}
\mathsf{write\_}(p_1, \epsilon, p_2) = (p_2 = p_1)\\
\mathsf{write\_}(p_1, c\cdot\overline{c}, p_3) = \exists p_2.\;\mathsf{write\_char\_}(p_1, c, p_2) \land \mathsf{write\_}(p_2, \overline{c}, p_3)
\end{array}$$
$$\annot{\{\mathbf{token}(p_1) \land \mathsf{write\_}(p_1, \overline{c}, p_2)\}}\ \mathsf{write}(\overline{c})\ \annot{\{\mathbf{token}(p_2)\}}$$
Note: we assume that the effect of $\mathsf{write}(\overline{c}\cdot\overline{c}')$ is indistinguisable from that of $\mathsf{write}(\overline{c}); \mathsf{write}(\overline{c}')$; to model this, we take $\mathsf{write\_char}$ as a primitive I/O tag and we define $\mathsf{write\_}$ in terms of it. Accordingly, in the formal setting we assume $\mathsf{write}$ is some function implemented in terms of $\mathsf{write\_char}(c)$ primitive I/O expressions.

Since system calls are expensive, $\mathsf{putchar}$ buffers output in a global variable:
$$\begin{array}{l}
\mathbf{function}\ \mathsf{putchar}(c)\ \{\\
\quad \annot{\{\mathsf{buffer} \mapsto \overline{c} * \mathbf{token}(p_1) \land \mathsf{write\_}(p_1, \overline{c}, p_2) \land \mathsf{write\_char\_}(p_2, c, p_3)\}}\\
\quad \mathbf{if}\ |!\mathsf{buffer}| = 1000\ \mathbf{then}\ \{\\
\quad\quad \mathsf{write}(!\mathsf{buffer}); \mathsf{buffer} \leftarrow \epsilon\\
\quad \}\\
\quad \mathsf{buffer} \leftarrow !\mathsf{buffer} \cdot c\\
\quad \annot{\{\exists \overline{c}', p_1'.\;\mathsf{buffer} \mapsto \overline{c}' * \mathbf{token}(p_1') \land \mathsf{write\_}(p_1', \overline{c}', p_3)\}}\\
\}
\end{array}$$
The specification for $\mathsf{putchar}$ that we show here does not hide the details of how it is implemented. For example, it names the global variable $\mathsf{buffer}$. We would like to define predicate $\mathsf{putchar\_}$ such that this implementation satisfies the simple, abstract specification for $\mathsf{putchar}$ that we showed on p.~\pageref{def:putchar}. Notice that this requires that we unify the postcondition above with $\mathbf{token}(q)$, for some place $q$. Clearly, in the approach of \S\ref{sec:petrinets}, that is impossible, since $\mathbf{token}(q)$ constrains only the marking, not the heap.

\subsection{Assertions as places}

What we want is a specification for $\mathsf{putchar}$ that looks exactly like the one on p.~\pageref{def:putchar}, and that therefore allows the client program specifications and proofs shown in \S\ref{sec:examples}, but which at the same time can be unified with the specification for the buffering implementation shown above.

The main contribution of this paper is the observation that we can achieve this by introducing, to complement the existing \emph{primitive} notion of places and the existing \emph{primitive} $\mathbf{token}$ assertion form, an \emph{abstract notion of places}, and an \emph{abstract version of the $\mathbf{token}$ assertion form}, where \emph{places are assertions}, and $\mathbf{token}(p)$ simply means $p$.

We can then unify the \emph{abstract reading} of the specification of $\mathsf{putchar}$ on p.~\pageref{def:putchar} with the one above by defining $\mathsf{putchar\_}$ as follows (where $p_1, p_2 \in \mathit{Asns}$):
$$\begin{array}{l}
\mathsf{putchar\_}(p_1, c, p_2) = \exists p_1', p_2'.\\
\quad p_1 = \mathsf{buffer\_token}(p_1') \land \mathsf{write\_char\_}(p_1', c, p_2') \land p_2 = \mathsf{buffer\_token}(p_2')\\
\textrm{where}\\
\mathsf{buffer\_token}(p) = \exists \overline{c}, p_0.\;\mathsf{buffer} \mapsto \overline{c} * \mathbf{token}(p_0) \land \mathsf{write\_}(p_0, \overline{c}, p)
\end{array}$$
Note: we will continue to write $\mathbf{token}(p)$ instead of just $p$ when we wish to point out that we are applying the Petri net specification style.

Notice that by adopting the abstract reading of the specification of p.~\pageref{def:putchar} as the specification of $\mathsf{putchar}$, we can build proofs of clients of $\mathsf{putchar}$ that are agnostic as to whether $\mathsf{putchar}$ is a primitive I/O tag, a function that performs composite I/O (i.e.~multiple I/O actions, as exemplified by the $\mathsf{putchars}$ example in \S\ref{sec:composite}), or even a function that performs both I/O and heap manipulation, such as the buffering implementation above.

Notice that in case $\mathsf{putchar}$ is a primitive I/O tag, the abstract reading of its specification is satisfied trivially by the primitive reading, if we define $t\_(p, v, v', q)$, where $t$ is a primitive I/O tag and $p$ and $q$ are assertions, to mean $\exists p', q' \in \mathcal{P}.\;p = \mathbf{token}(p') \land t\_(p', v, v', q') \land q = \mathbf{token}(q')$. 

\subsection{Composing I/O specifications}

Consider the following program:
$$\begin{array}{l}
\mathbf{function}\ \mathsf{start}()\ \{\\
\quad \annot{\{\mathbf{token}(p_1) \land \mathsf{beep\_}(p_1, p_2) \land \mathsf{write\_}(p_2, \texttt{'!'}, p_3)\}}\\
\quad \\
\quad \mathbf{let}\ \mathsf{buffer} := \mathbf{ref}\ \epsilon\ \mathbf{in}\\
\quad \mathbf{function}\ \mathsf{putchar}(c)\ \{\ \cdots\ \}\\
\quad \mathbf{function}\ \mathsf{flush}()\ \{\ \cdots\ \}\\
\quad \\
\quad \mathbf{function}\ \mathsf{main}()\ \{\\
\quad \quad \annot{\{\mathbf{token}(q_1) \land \mathsf{beep\_}(q_1, q_2) \land \mathsf{putchar\_}(q_2, \texttt{'!'}, q_3) \land \mathsf{flush\_}(q_3, q_4)\}}\\
\quad \quad \mathsf{beep}();\ \mathsf{putchar}(\texttt{'!'});\ \mathsf{flush}()\\
\quad \quad \annot{\{\mathbf{token}(q_4)\}}\\
\quad \}\\
\quad \mathsf{main}()\\
\quad \\
\quad \annot{\{\mathbf{token}(p_3)\}}\\
\}
\end{array}$$
Function $\mathsf{main}$ is written in terms of the C run-time library functions $\mathsf{putchar}$ and $\mathsf{flush}$, as well as the system call $\mathsf{beep}$. We would like its specification and verification to be independent of the implementation of $\mathsf{putchar}$ and $\mathsf{flush}$ in terms of the buffer. Function $\mathsf{start}$, which first initializes the C run-time library's internal data structures and then calls $\mathsf{main}$, is specified purely in terms of system calls.

Function $\mathsf{main}$'s specification shown above is of the form $\{P(q_1, q_2, q_3, q_4)\}$ $\mathsf{main}()$ $\{Q(q_4)\}$. We can prove easily that (the abstract reading of) this Hoare triple holds for arbitrary values of $q_1$, $q_2$, $q_3$, and $q_4$. This means that when verifying a particular call of $\mathsf{main}$, we can instantiate $\mathsf{main}$'s specification with arbitrary particular values for these variables.

To verify the call of $\mathsf{main}$ in $\mathsf{start}$, we need to prove the following Hoare triple:
$$\begin{array}{l}
\annot{\{\mathsf{buffer}\mapsto\epsilon * \mathbf{token}(p_1) \land \mathsf{beep\_}(p_1, p_2) \land \mathsf{write\_}(p_2, \texttt{'!'}, p_3)\}}\\
\mathsf{main}()\\
\annot{\{\mathbf{token}(p_3)\}}
\end{array}$$
Specifically, we need to find values for $q_1$, $q_2$, $q_3$, and $q_4$ such that the two implications
$$\begin{array}{c}
\mathsf{buffer}\mapsto\epsilon * \mathbf{token}(p_1) \land \mathsf{beep\_}(p_1, p_2) \land \mathsf{write\_}(p_2, \texttt{'!'}, p_3)\\
\Downarrow\\
\mathbf{token}(q_1) \land \mathsf{beep\_}(q_1, q_2) \land \mathsf{putchar\_}(q_2, \texttt{'!'}, q_3) \land \mathsf{flush\_}(q_3, q_4)
\end{array}$$
and
$$\mathbf{token}(q_4) \Rightarrow \mathbf{token}(p_3)$$
hold.

Proving the first implication requires that predicate $\mathsf{beep\_}$ satisfy a \emph{frame property}:
$$\mathsf{beep\_}(p_1, p_2) \Rightarrow \mathsf{beep\_}(p_1 * R, p_2 * R)$$
Therefore, it is useful when applying this approach to introduce a convention to have such \emph{abstract transition predicates} satisfy the frame property, as well as the following \emph{weakening property}, similar to Hoare logic's Rule of Consequence:
$$\inferrule{
p_1' \sqsubseteq p_1\\
p_2 \sqsubseteq p_2'
}{
\mathsf{beep}(p_1, p_2) \Rightarrow \mathsf{beep}(p_1', p_2')
}$$

To comply with this convention, we redefine predicate $\mathsf{putchar\_}$ as follows:
$$\begin{array}{l}
\mathsf{putchar\_}(p_1, c, p_2) = \exists p_1', p_2', R.\\
\quad (p_1 \sqsubseteq \mathsf{buffer\_token}(p_1') * R) \land \mathsf{write\_char\_}(p_1', c, p_2') \land (\mathsf{buffer\_token}(p_2') * R \sqsubseteq p_2)
\end{array}$$
Furthermore, we define $\mathsf{flush\_}$ as follows:
$$\mathsf{flush\_}(p_1, p_2) = \exists p', R.\;
(p_1 \sqsubseteq \mathsf{buffer\_token}(p') * R) \land (\mathsf{buffer}\mapsto \epsilon * p' * R \sqsubseteq p_2)$$

We can then verify the call of $\mathsf{main}$ by taking $q_1 = p_1 * \mathsf{buffer}\mapsto\epsilon$ and $q_2 = p_2 * \mathsf{buffer}\mapsto\epsilon$ and $q_3 = \mathsf{buffer\_token}(p_3)$ and $q_4 = \mathsf{buffer}\mapsto\epsilon * p_3$. Notice that $q_2 \Rightarrow \mathsf{buffer\_token}(p_2)$, as required by $\mathsf{putchar\_}(q_2, \texttt{'!'}, q_3)$.

We have verified this example using Iris \cite{iris-io-2-0}.

\subsection{Rewriting I/O specifications}

We define the abstract reading of the special transition predicates $\mathbf{split\_}$, $\mathbf{join\_}$, and $\mathbf{noop\_}$ as follows:
$$\begin{array}{r @{\quad=\quad} l}
\mathbf{split\_}(p_1, p_2, p_3) & (p_1 \sqsubseteq p_2 * p_3)\\
\mathbf{join\_}(p_1, p_2, p_3) & (p_1 * p_2 \sqsubseteq p_3)\\
\mathbf{noop\_}(p_1, p_2) & (p_1 \sqsubseteq p_2)
\end{array}$$

It follows that in the abstract reading we have many equivalences between I/O specifications which are not available in the primitive reading. For example:
$$(\exists p'.\;\mathbf{split\_}(p, q_1, p') \land \mathbf{split\_}(p', q_2, q_3)) \Leftrightarrow (\exists p'.\ \mathbf{split\_}(p, p', q_3) \land \mathbf{split\_}(p', q_1, q_2))$$
Here, too, the continued use of the $\mathbf{split\_}$, $\mathbf{join\_}$, and $\mathbf{noop\_}$ syntax is a purely stylistic choice.

\subsection{Transition predicates as \emph{abstract nested Hoare triples}}

Readers will have noted the similarity between our abstract transition predicates, such as $\mathsf{putchar\_}(P, c, Q)$, and Hoare triples $\{P\}\ \mathsf{putchar}(c)\ \{Q\}$. Both specify a precondition and a postcondition for an action, and both satisfy the Frame rule and the Rule of Consequence. Both can be used as first-class assertions in modern program logics that support \emph{nested Hoare triples} (e.g.~\cite{nested-hoare-triples} and Iris \cite{iris3}).

However, note also the differences: transition predicates are more abstract, in that they need not correspond to a particular function or program expression. (See e.g.~the $\mathsf{write\_char\_}$ transition predicate introduced at the start of this section to model the effect of the $\mathsf{write}$ system call.) Furthermore, transition predicates can be \emph{folded} and \emph{unfolded} by modules that have access to their definition. Thirdly, whereas for output actions such as $\mathsf{putchar}$ nested Hoare triples could mostly be used instead of transition predicates in our I/O specifications, for input actions such as $\mathsf{getchar}$ this is less straightforward.

\section{Dealing with input using prophecy variables}\label{sec:dealing-with-input}

\subsection{A motivating example}

Consider the following specification for a chat server. For simplicity, we assume there is a single chat room, with exactly two members, with nicknames $n_1$ and $n_2$.

$$\begin{array}{l}
\annot{\left\{\begin{array}{l l}
& (\mathbf{token}(r_1) \land \mathsf{receiveFromNick\_}(r_1, n_1, \mu_1))\\
{}* & (\mathbf{token}(r_2) \land \mathsf{receiveFromNick\_}(r_2, n_2, \mu_2))\\
{}* & (\mathbf{token}(s) \land (\forall \mu \in \mu_1^{n_1}\ ||\ \mu_2^{n_2}.\;\mathsf{sendToNicks\_}(s, \mu)))
\end{array}\right\}}\\
\mathsf{serveChatRoom}()\\
\annot{\{\mathbf{false}\}}\\
\end{array}$$
where
$$\begin{array}{@{} l @{}}
\mathsf{sendToNicks\_}(s, \mu) =\\
\quad \exists s_1, s_2.\;\mathbf{split\_}(s, s_1, s_2) \land \mathsf{sendToNick\_}(s_1, n_1, \mu) \land \mathsf{sendToNick\_}(s_2, n_2, \mu)
\end{array}$$

We use symbol $\mu$ and variants to range over infinite sequences of messages, $m^n$ to denote the quoted message \texttt{$n$ says '$m$'}, lifted also to sequences of messages, and $\mu_1\ ||\ \mu_2$ to denote the set of all interleavings of $\mu_1$ and $\mu_2$.

The specification states that the chat server sends to each member the same sequence of messages $\mu$, which is some interleaving of the sequences of messages received from each member. Notice that in this Petri net, there are many $\mathbf{split}$ transitions outgoing from place $s$, but since there is only a single token in this place, only one of these transitions can fire in any particular execution. This models the fact that the chat server must choose which interleaving of the incoming messages it will send out.

For simplicity, we assume that the following network API is available to the chat server:
$$\begin{array}{l}
\annot{\{\mathbf{token}(r) \land \mathsf{receiveFromNick\_}(r, n, m\cdot\mu)\}}\\
\mathsf{receiveFromNick}(n)\\
\annot{\{\exists r'.\;\mathbf{token}(r') \land \mathsf{receiveFromNick\_}(r', n, \mu) \land \mathsf{res} = m\}}\\
\\
\annot{\{\mathbf{token}(s) \land \mathsf{sendToNick\_}(s, n, m\cdot\mu)\}}\\
\mathsf{sendToNick}(n, m)\\
\annot{\{\exists s'.\;\mathbf{token}(s') \land \mathsf{sendToNick\_}(s', n, \mu)\}}
\end{array}$$

We implement the chat server by forking one thread per member to receive messages from that member and insert them into a shared queue or \emph{channel}, and, in a separate thread, dequeuing messages from the channel and sending them to each member; see Fig.~\ref{fig:chat-impl}.
\begin{figure}
$$\begin{array}{l}
\mathbf{function}\ \mathsf{serveChatRoom}()\ \{\\
\quad \mathbf{let}\ \mathsf{roomChan} := \mathsf{newChannel}()\ \mathbf{in}\\
\quad \mathbf{fork}\ \mathsf{pumpFromNick}(n_1, \mathsf{roomChan});\\
\quad \mathbf{fork}\ \mathsf{pumpFromNick}(n_2, \mathsf{roomChan});\\
\quad \mathsf{pumpRoom}(\mathsf{roomChan})\\
\}\\
\\
\mathbf{function}\ \mathsf{pumpFromNick}(n, \mathsf{roomChan})\ \{\\
\quad \annot{\left\{\begin{array}{l}
(\mathbf{token}(r) \land \mathsf{receiveFromNick\_}(r, n, \mu))\\
{} * (\mathbf{token}(s) \land \mathsf{send\_}(s, \mathsf{roomChan}, \mu^n))
\end{array}\right\}}\\
\quad \mathbf{loop}\ \{\\
\quad\quad \mathbf{let}\ m := \mathsf{receiveFromNick}(n)\ \mathbf{in}\\
\quad\quad \mathsf{send}(\mathsf{roomChan}, m^n)\\
\quad \}\\
\quad \annot{\{\mathbf{false}\}}\\
\}\\
\\
\mathbf{function}\ \mathsf{pumpRoom}(\mathsf{roomChan})\ \{\\
\quad \annot{\left\{\begin{array}{l}
(\mathbf{token}(r) \land \mathsf{receive\_}(r, \mathsf{roomChan}, \mu))\\
{} * (\mathbf{token}(s_1) \land \mathsf{sendToNick\_}(s_1, n_1, \mu))\\
{} * (\mathbf{token}(s_2) \land \mathsf{sendToNick\_}(s_2, n_2, \mu))
\end{array}\right\}}\\
\quad \mathbf{loop}\ \{\\
\quad\quad \mathbf{let}\ m := \mathsf{receive}(\mathsf{roomChan})\ \mathbf{in}\\
\quad\quad \mathsf{sendToNick}(n_1, m); \mathsf{sendToNick}(n_2, m)\\
\quad \}\\
\quad \annot{\{\mathbf{false}\}}\\
\}
\end{array}$$
\caption{Chat server implementation}\label{fig:chat-impl}
\end{figure}

We wish to specify and verify each of these threads in a way that abstracts over the fact that all threads are running in the same process and communicating through an in-process shared queue. Indeed, as far as functions $\mathsf{pumpFromNick}$ and $\mathsf{pumpRoom}$ are concerned, the channel might as well be an inter-process or inter-machine communication construct, so it makes sense that we specify sending and receiving on the channel exactly analogously to the network APIs $\mathsf{sendToNick}$ and $\mathsf{receiveFromNick}$:
$$\begin{array}{l}
\annot{\{\mathbf{token}(s) \land \mathsf{send\_}(s, c, m\cdot\mu)\}}\\
\mathsf{send}(c, m)\\
\annot{\{\exists s'.\;\mathbf{token}(s') \land \mathsf{send\_}(s', c, \mu)\}}\\
\\
\annot{\{\mathbf{token}(r) \land \mathsf{receive\_}(r, c, m\cdot\mu)\}}\\
\mathsf{receive}(c)\\
\annot{\{\exists r'.\;\mathbf{token}(r') \land \mathsf{receive\_}(r', c, \mu) \land \mathsf{res} = m\}}
\end{array}$$
Verifying the implementations of functions $\mathsf{pumpFromNick}$ and $\mathsf{pumpRoom}$ against their specifications is straightforward. Verifying the main function $\mathsf{serveChatRoom}$ is straightforward as well, provided that we may assume the following specification for function $\mathsf{newChannel}$:\footnote{For simplicity, this specification is specialized for the case of two senders and one receiver.}
$$\begin{array}{l}
\annot{\{\mathbf{true}\}}\\
\mathsf{newChannel}()\\
\annot{\left\{\begin{array}{@{} l @{}}
(\mathbf{token}(s_1) \land \mathsf{send\_}(s_1, \mathsf{res}, \mu_1))\\
{} * (\mathbf{token}(s_2) \land \mathsf{send\_}(s_2, \mathsf{res}, \mu_2))\\
{} * (\exists \mu \in \mu_1\ ||\ \mu_2.\;\mathbf{token}(r) \land \mathsf{receive\_}(r, \mathsf{res}, \mu))
\end{array}\right\}}
\end{array}$$

\section{A programming language with prophecy variables}\label{sec:prophvars}

Notice that verifying an implementation of the channel construct, for example in terms of a shared queue, against these specifications is not possible using a straightforward application of our Hoare logic from \S\ref{sec:hoare-logic}, or other existing program logics such as Iris \cite{iris3}. Indeed, in these logics, the precondition of $\mathsf{receive}$ constrains only the pre-state of a call of $\mathsf{receive}$. Since, starting from this same pre-state, many different thread schedulings, and, consequently, many different return values of $\mathsf{receive}$ are generally possible, there is no relationship between the pre-state and the result value so the precondition cannot express such a relationship. A solution to this problem has been long known, however; it is known as \emph{prophecy variables} \cite{abadi-lamport,Zhang2012}.

In order to apply our abstract I/O specification and verification approach to in-memory input constructs such as channels in the context of logics such as Iris, then, we propose to apply this well-known idea of prophecy variables. In particular, we propose to apply such logics not directly to the actual programming language and program involved, but to a version of the programming language and the program \emph{instrumented} with prophecy variables. The end-to-end approach for verifying a program, then, is to first verify correctness of the instrumented version using a logic like Iris, and then to apply an \emph{erasure} theorem that maps this correctness property to a corresponding property of the original program.

In this paper, we assume that the correctness property of interest can be expressed as the program not getting stuck. This is true whenever the property can be translated into run-time checks inserted into the program. The erasure theorem needed, then, is simply that if no instrumented execution gets stuck, then no erased execution gets stuck.

In the remainder of this section, we first elaborate this idea for simple prophecy variables that are assigned an arbitrary value once. We then discuss prophecy variables to which a \emph{sequence} of values is assigned incrementally. Next, we show how to build \emph{constrained} prophecy variables on top of these in the logic. In the next sections, we show how these prophecy variables can be used to verify a channel implementation against the abstract I/O specifications proposed above.

\subsection{Simple prophecy variables}

Simple prophecy variables can be added to a programming language by adding a type of \emph{prophecy variable identifiers}, a command $\mathbf{create\_pvar}()$ for allocating and returning a prophecy variable identifier and associating with it a \emph{prophecy value}, and a command $\mathbf{assign\_pvar}(\iota, v)$, which assigns value $v$ to the prophecy variable with identifier $\iota$. After the assignment operation completes, we have that the assigned value equals the prophecy value.

We can think of these operations operationally as follows: the creation operation picks an arbitrary prophecy value nondeterministically; the assignment operation either does nothing if the assigned value equals the prophecy value, or enters an infinite loop otherwise.

We formalize this as follows. Suppose the base language's semantics is defined using a small-step relation $(\sigma, e) \rightarrow (\sigma', e')$, where $\sigma$ ranges over states and $e$ over expressions, including a rule
$$\inferrule{
(\sigma, e) \hookrightarrow (\sigma', e')
}{
(\sigma, K[e/\bullet]) \rightarrow (\sigma', K[e'/\bullet])
}$$
for lifting step rules over evaluation contexts $K$. An example of such a language is our programming language from \S\ref{sec:lang-with-io} and its monitoring semantics defined in \S\ref{sec:monitoring}, where states $\sigma = (T, h)$ consist of a set of I/O traces $T$ and a heap $h$. We extend the definitions of expressions, values, and evaluation contexts as follows:
$$\begin{array}{r l}
& \iota \in \mathit{PVarIdents}\\
e ::= & \cdots\ |\ \mathbf{create\_pvar}()\ |\ \mathbf{assign\_pvar}(e, e)\\
v ::= & \cdots\ |\ \iota\\
K ::= & \cdots\ |\ \mathbf{assign\_pvar}(\bullet, e)\ |\ \mathbf{assign\_pvar}(\iota, \bullet)
\end{array}$$
We extend states with a \emph{prophecy heap} $\rho : \mathit{PVarIdents} \rightharpoonup_\mathrm{fin} \mathit{Vals}$, a finite partial function from prophecy variable identifiers to values. Existing commands leave the prophecy heap unchanged; the step rules for the new commands are as follows:
\begin{mathpar}
\inferrule[CreatePVar]{
\iota \notin \mathrm{dom}(\rho)
}{
(\sigma, \rho, \mathbf{create\_pvar}()) \hookrightarrow (\sigma, \rho[\iota:=v], \iota)
}
\and
\inferrule[AssignPVarMatch]{
\iota \in \mathrm{dom}(\rho)\\
\rho(\iota) = v
}{
(\sigma, \rho, \mathbf{assign\_pvar}(\iota, v)) \hookrightarrow (\sigma, \rho[\iota:=\bot], ())
}
\and
\inferrule[AssignPVarNoMatch]{
\iota \in \mathrm{dom}(\rho)\\
\rho(\iota) \neq v
}{
(\sigma, \rho, \mathbf{assign\_pvar}(\iota, v)) \hookrightarrow (\sigma, \rho, \mathbf{assign\_pvar}(\iota, v))
}
\end{mathpar}
Notice that successfully assigning a prophecy variable removes it from the prophecy heap.

From this semantics, we can derive the following Hoare logic proof rules for the new commands:
$$\begin{array}{c}
\annot{\{\mathbf{true}\}}\ \mathbf{create\_pvar}()\ \annot{\{\exists v.\;\mathsf{pvar}(\mathsf{res}, v)\}}\\
\annot{\{\mathsf{pvar}(\iota, v)\}}\ \mathbf{assign\_pvar}(\iota, v')\ \annot{\{v = v'\}}
\end{array}$$
where the assertion $\mathsf{pvar}(\iota, v) = \{(\sigma, \rho)\ |\ \iota \in \mathrm{dom}(\rho) \land \rho(\iota) = v\}$ denotes the existence of a prophecy variable with identifier $\iota$ and prophecy value $v$.

For verification using prophecy variables to be sound, we need the property that if a program instrumented with prophecy variables does not get stuck, then the original program also does not get stuck. This is not immediately obvious, because of the \textsc{AssignPVarNoMatch} rule. Therefore, we first prove erasure to an intermediate semantics for instrumented programs that still tracks prophecy variables (in particular, it tracks the set $I$ of allocated prophecy variable identifiers) but that does not have the looping behavior:
\begin{mathpar}
\inferrule[CreatePVar-I]{
\iota \notin I
}{
(\sigma, I, \mathbf{create\_pvar}()) \hookrightarrow_\mathrm{I} (\sigma, I \cup \{\iota\}, \iota)
}
\and
\inferrule[AssignPVar-I]{
\iota \in I
}{
(\sigma, I, \mathbf{assign\_pvar}(\iota, v)) \hookrightarrow_\mathrm{I} (\sigma, I \setminus\{\iota\}, ())
}
\end{mathpar}
The essential property of the instrumented semantics is that if a configuration is reachable in the intermediate semantics, then \emph{every} corresponding configuration, i.e.~with every possible assignment of prophecy values to the allocated prophecy variables, is reachable in the instrumented semantics:
\begin{lemma}
If $(\sigma_0, \emptyset, e_0) \rightarrow_\mathrm{I}^* (\sigma, I, e)$ and $\mathrm{dom}(\rho) = I$ then $(\sigma_0, \emptyset, e_0) \rightarrow^* (\sigma, \rho, e)$.
\end{lemma}
\begin{proof}
By induction on the number of steps and case analysis on the step rule. For rule \textsc{AssignPVar-I}, the induction hypothesis guarantees that a configuration with a prophecy value that matches the assigned value is reachable.
\end{proof}
We can now prove that if a configuration is safe under the instrumented semantics, then it is safe under the intermediate semantics:
\begin{theorem}
If $\mathrm{safe}((\sigma_0, \emptyset, e_0))$ then $\mathrm{safe}_\mathrm{I}((\sigma_0, \emptyset, e_0))$.
\end{theorem}
\begin{proof}
Assume $(\sigma_0, \emptyset, e_0) \rightarrow^*_\mathrm{I} (\sigma, I, e)$. Pick some arbitrary $\rho$ such that $\mathrm{dom}(\rho) = I$. By the lemma above, we have that $(\sigma_0, \emptyset, e_0) \rightarrow^* (\sigma, \rho, e)$. By the premise we have that this configuration can make a step. By case analysis on the step rule, it is easy to prove that the intermediate semantics can make a similar step. In particular, if the instrumented semantics makes an \textsc{AssignPVarNoMatch} step, the intermediate semantics can make an \textsc{AssignPVar-I} step. (This theorem does not claim that the resulting configurations correspond in any way.)
\end{proof}
A machine-checked version of this development is available cite{iris-io-2-0}.

Notice, now, that if an instrumented program does not get stuck under the intermediate semantics, then for this program the $\mathbf{assign\_pvar}$ command is equivalent to $\mathbf{skip}$. If furthermore we assume that no other constructs of the programming language allow inspection of a prophecy variable identifier (as is the case for the instrumented version of our programming language from \S\ref{sec:lang-with-io}), we can conclude that the type of prophecy variable identifiers is equivalent to the unit type and $\mathbf{create\_pvar}$ is equivalent to the unit value literal $()$.

\subsection{Incremental prophecy variables}

Above, we introduced a simple form of prophecy variables, which are assigned once. However, notice that in the case of the channels example, the prophecy value $\mu$ is a sequence where each element corresponds to a different receive operation, occurring at a different point in the execution of the program. To support this, we here propose a variant of prophecy variables where the prophecy value predicts a \emph{sequence} of assigned values.

The syntax of instrumented programs does not change.

We update the definition of prophecy heaps to map identifiers to sequences of values:
$$\rho : \mathit{PVarIdents} \rightharpoonup_\mathrm{fin} \mathit{Values}^\omega$$

We update the step rules as follows:
\begin{mathpar}
\inferrule[CreatePVar]{
\iota \notin \mathrm{dom}(\rho)
}{
(\sigma, \rho, \mathbf{create\_pvar}()) \hookrightarrow (\sigma, \rho[\iota:=\mu], \iota)
}
\and
\inferrule[AssignPVarMatch]{
\iota \in \mathrm{dom}(\rho)\\
\rho(\iota) = v\cdot\mu
}{
(\sigma, \rho, \mathbf{assign\_pvar}(\iota, v)) \hookrightarrow (\sigma, \rho[\iota:=\mu], ())
}
\and
\inferrule[AssignPVarNoMatch]{
\iota \in \mathrm{dom}(\rho)\\
\rho(\iota) = v'\cdot\mu\\
v \neq v'
}{
(\sigma, \rho, \mathbf{assign\_pvar}(\iota, v)) \hookrightarrow (\sigma, \rho, \mathbf{assign\_pvar}(\iota, v))
}
\end{mathpar}
Notice that after a successful assignment, the assigned value is popped from the front of the prophecy value in the prophecy heap.

This semantics allows us to derive the following Hoare rules:
$$\begin{array}{c}
\annot{\{\mathbf{true}\}}\ \mathbf{create\_pvar}()\ \annot{\{\exists \mu.\;\mathsf{pvar}(\mathsf{res}, \mu)\}}\\
\annot{\{\mathsf{pvar}(\iota, v\cdot\mu)\}}\ \mathbf{assign\_pvar}(\iota, v')\ \annot{\{\mathsf{pvar}(\iota, \mu) \land v = v'\}}
\end{array}$$
The intermediate semantics remains unchanged, and the erasure proof proceeds completely analogously.

\subsection{Constrained incremental prophecy variables}

Above, prophecy variable creation produces a completely arbitrary prophecy value. However, in the channels example we need to know at the point of channel creation that the sequence of received values will be an interleaving of the sequences of sent values. Therefore, we here propose a third type of prophecy variables, called \emph{constrained incremental prophecy variables}, that allow the specification, at prophecy variable creation time, of a \emph{constraint} on the sequence of assigned values. This constraint is enforced at assignment time, but we have, already at creation time, that the prophecy value satisfies the constraint. An important side condition for this to be sound is of course that the constraint be satisfiable.

We can in fact build this feature in the logic, on top of the instrumented semantics for unconstrained incremental prophecy variables proposed above. Specifically, we can prove the following Hoare proof rules, where $\mathsf{cpvar}(\iota, M, \mu)$ denotes a constrained prophecy variable where the sequence of assigned values is constrained to be in the set of sequences $M$.
$$\begin{array}{c}
\annot{\{M \neq \emptyset\}}\ \mathbf{create\_pvar}(M)\ \annot{\{\exists \mu.\;\mathsf{cpvar}(\mathsf{res}, M, \mu) \land \mu \in M\}}\\
\annot{\{\mathsf{cpvar}(\iota, M, v\cdot\mu) \land M[v'] \neq \emptyset\}}\  \mathbf{assign\_pvar}(\iota, v')\ \annot{\{\mathsf{cpvar}(\iota, M[v'], \mu) \land v = v'\}}
\end{array}$$
where $M[v] = \{\mu\ |\ v\cdot\mu \in M\}$.

When implementing constrained prophecy variable creation on top of unconstrained prophecy variable creation, there are two cases to consider after creating the unconstrained prophecy variable: either the prophecy value satisfies the constraint, or it does not. Our strategy, then, is as follows: if it does, expose this value as the prophecy value of the constrained prophecy variable. Otherwise, pick an arbitrary value that does satisfy the constraint, and expose that as the prophecy value of the constrained prophecy variable. At prophecy variable assignment time, it will turn out that the unconstrained prophecy value did satisfy the constraint after all.

To apply this idea to the incremental case, at prophecy variable creation time, intuitively we pick as the constrained prophecy value the sequence that satisfies the constraint and that \emph{maximally} matches the unconstrained prophecy value, i.e.~that has a maximal-length prefix that matches the unconstrained prophecy value. As assignments occur, it will gradually become clear that the whole unconstrained prophecy value did satisfy the constraint after all.

However, such a maximally-matching sequence does not necessarily exist. Indeed, consider the constraint $M_{\mathtt{A}\mathtt{B}} = \{\mathtt{A}^n\cdot\mathtt{B}^\omega\ |\ n \in \mathbb{N}\}$, that is, the set of all sequences consisting of a finite number of $\mathtt{A}$s followed by an infinite number of $\mathtt{B}$s. Now, consider the unconstrained prophecy value $\mu = \mathtt{A}^\omega$. There is no element of $M_{\mathtt{A}\mathtt{B}}$ that maximally matches $\mu$.

Note, however, that the proof rules proposed above do not enforce that the program adheres to a constraint such as $M_{\mathtt{A}\mathtt{B}}$. Indeed, a program that assigns $\mathtt{A}$s indefinitely can be verified using these proof rules. Still, note also that all finite prefixes of the execution of this program do adhere to the constraint.

Therefore, in this paper we restrict ourselves to \emph{partial correctness} verification, i.e.~verification of safety properties, rather than liveness properties. This means that it suffices to consider only each finite prefix of the executions of the program.

In this restricted setting, we can in fact solve the problem. We exploit the fact that if the execution prefix we are currently considering has length $n$, then at most $n$ prophecy variable assignments can occur in this execution. It follows that the elements of the constrained prophecy value at position $n$ and after will never be ``tested'' and therefore need not match the unconstrained prophecy value. Therefore, it suffices to pick a constrained prophecy value that maximally matches the unconstrained one \emph{up to the length of the execution prefix}.

We can encode this in Iris, a logic for partial correctness verification, using its \emph{later} operator: $\triangleright P$ (pronounced \emph{later $P$}) holds for execution prefixes of length $n$ if $P$ holds for execution prefixes of length $n' < n$. In particular, when considering an execution prefix of length 0, we have $\triangleright \mathsf{False}$. We will add a later operator to our Hoare logic in the next section.

Using the later operator, we can define the set $\mathsf{maxmatch}(\mu, M)$ of elements of $M$ that match $\mu$ maximally up to the length of the execution prefix, as follows:
$$\mathsf{maxmatch}(v\cdot\mu, M) = \left\{\begin{array}{@{} l l}
M & \textrm{if $M[v] = \emptyset$}\\
\{v\cdot\mu' \in M\ |\ \triangleright (\mu' \in \mathsf{maxmatch}(\mu, M[v]))\} & \textrm{otherwise}
\end{array}\right.$$

We can then define $$\mathsf{cpvar}(\iota, M, \mu) = \exists \mu'.\;\mathsf{pvar}(\iota, \mu') \land \mu \in \mathsf{maxmatch}(\mu', M)$$ From this definition, and the Hoare rules for incremental prophecy variables seen above, we can easily prove the proposed Hoare rules for constrained incremental prophecy variables.

We have developed a machine-checked version of constrained incremental prophecy variables in Coq and integrated it with the Iris logic \cite{iris-io-2-0}.

\section{A programming language with I/O, prophecy variables, and concurrency}\label{sec:concur}

In this section, we extend our programming language from \S\ref{sec:lang-with-io} with a $\mathbf{fork}$ command, an atomic load command $\langle !\ell\rangle$ and a compare-and-set command $\mathsf{CAS}(\ell, v_\mathrm{old}, v_\mathrm{new})$ (\S\ref{sec:concur-sem}), and we extend our Hoare logic with support for concurrency (including ghost cells, shared regions and shared region invariants) and a later operator (\S\ref{sec:concur-hoare}).

\subsection{Programming language syntax and semantics}\label{sec:concur-sem}

Starting from the programming language of \S\ref{sec:labeled}, we extend the syntax of expressions and, correspondingly, the syntax of evaluation contexts, as follows:
$$\begin{array}{r @{\ } l}
e ::= & \cdots\ |\ \mathbf{fork}(e)\ |\ \langle !e\rangle\ |\ \mathbf{CAS}(e, e, e)\\
K ::= & \cdots\ |\ \mathbf{fork}(K)\ |\ \langle !K\rangle\ |\ \mathbf{CAS}(K, e, e)\ |\ \mathbf{CAS}(v, K, e)\ |\ \mathbf{CAS}(v, v, K)
\end{array}$$
We update the definition of configurations to include a \emph{thread pool}, which is a list of expressions: $\gamma \in \mathit{Configs} = \mathit{States} \times \mathit{Exprs}^*$ where $\mathit{States} = \mathit{Heaps}$.

The labeled head reduction relation $\stackrel{\tau}{\hookrightarrow}$ now relates a pre-state and a pre-expression to a post-state, a post-expression, and a list of forked expressions. The step rules for the existing expressions are unchanged from \S\ref{sec:labeled}, except to specify that these expressions fork no threads. The new proof rules are as follows, where $=_\mathrm{dec}$ and $\neq_\mathrm{dec}$ are defined only on $\mathit{DataVals}$, the set of values constructed from $()$, $\mathbf{inl}$, $\mathbf{inr}$, and $({-}, {-})$:
\begin{mathpar}
h, \mathbf{fork}(e) \stackrel{\epsilon}{\hookrightarrow} h, (), e
\and
\inferrule{
\ell \in \mathrm{dom}(h)
}{
h, \langle !\ell\rangle \stackrel{\epsilon}{\hookrightarrow} h, h(\ell), \epsilon
}
\and
\inferrule{
\ell \in \mathrm{dom}(h)\\
h(\ell) =_\mathrm{dec} v_\mathrm{old}
}{
h, \mathbf{CAS}(\ell, v_\mathrm{old}, v_\mathrm{new}) \stackrel{\epsilon}{\hookrightarrow} h[\ell:=v], \mathbf{true}, \epsilon
}
\and
\inferrule{
\ell \in \mathrm{dom}(h)\\
h(\ell) \neq_\mathrm{dec} v_\mathrm{old}
}{
h, \mathbf{CAS}(\ell, v_\mathrm{old}, v_\mathrm{new}) \stackrel{\epsilon}{\hookrightarrow} h, \mathbf{false}, \epsilon
}
\and
\inferrule{
h, e \stackrel{\tau}{\hookrightarrow} h', e', \overline{e}_\mathrm{f}
}{
h, \overline{e} \cdot K[e/\bullet] \cdot \overline{e}' \stackrel{\tau}{\rightarrow} h', \overline{e}\cdot K[e'/\bullet] \cdot \overline{e}_\mathrm{f} \cdot \overline{e}'
}
\end{mathpar}

The definition of satisfaction of an I/O specification $T$ by a configuration given in \S\ref{sec:iospecs} applies unchanged.

We update the definition of a failed configuration as follows: a configuration has failed if any thread is not finished and not reducible:
$$\mathsf{failed}(h, \overline{e}) = \exists e \in \overline{e}.\;e \notin \mathit{Vals} \land \lnot\exists K, e', \tau, h', e'', \overline{e}_\mathrm{f}.\; e = K[e'/\bullet] \land h, e' \stackrel{\tau}{\hookrightarrow} h', e'', \overline{e}_\mathrm{f}$$

From this labeled semantics, we can construct a unlabeled, monitoring version, with configurations that include an I/O specification in the form of a prefix-closed set of traces, entirely analogously to how we did this for the single-threaded language in \S\ref{sec:monitoring}, in such a way that we obtain the following theorem:

\begin{theorem}
If $\mathsf{safe}(T, h, \overline{e})$ then $h, \overline{e} \vDash T$.
\end{theorem}

Starting from this concurrent programming language and its monitoring semantics, we can then obtain an instrumented language with I/O, concurrency, and incremental prophecy variables, and both an instrumented and an intermediate semantics for it, exactly as described in \S\ref{sec:prophvars}, that satisfy the erasure theorem. Configurations $\gamma = (T, h, \rho, \overline{e})$ of the instrumented semantics consist of an I/O specification $T$ (a prefix-closed set of I/O traces), a heap $h$, a prophecy heap $\rho$, and a thread pool $\overline{e}$ (a list of expressions).

\subsection{Hoare logic}\label{sec:concur-hoare}

We here present a minimal Hoare logic that is sufficient to verify the chat server example, including the channels implementation. For a more complete logic, we refer to the literature (e.g.~\citeN{iris3}).
To reason about concurrency, our Hoare logic includes ghost cells, fractional permissions, and shared regions with (first-order) shared region invariants.

We define the set of \emph{chunks} $\alpha$ as follows, where $\hat{\ell} \in \mathit{GhostLocs}$ ranges over an infinite set of ghost locations and $\hat{v} \in \mathit{GhostVals}$ ranges over a set of ghost values, which can be picked arbitrarily for a given proof:
$$\alpha \in \mathit{Chunks} ::= \ell \mapsto v\ |\ \hat{\ell} \mapsto \hat{v}\ |\ \mathsf{pvar}(\iota, \mu)\ |\ \mathsf{token}(p)$$
We define the \emph{logical heaps} $H \in \mathit{LogHeaps} = \mathit{Chunks} \rightarrow \mathbb{R}^+$ as the functions from chunks to \emph{fractions}, which are nonnegative reals. We lift addition on reals pointwise to logical heaps: $H + H' = \lambda \alpha.\;H(\alpha) + H'(\alpha)$. Notation $\mathbf{0} = \lambda \_.\;0$ denotes the empty logical heap; $\llbrace\alpha\rrbrace = \mathbf{0}[\alpha:=1]$ denotes the logical heap that contains chunk $\alpha$ with fraction 1. We can interpret a marking $V$ as the logical heap $H$ where $H(\alpha) = V(p)$ if $\alpha = \mathsf{token}(p)$, and $H(\alpha) = 0$ otherwise. Similarly, we can interpret a heap $h$ as the logical heap $H$ where $H(\alpha) = 1$ if $\alpha = \ell \mapsto v$ and $\ell \in \mathrm{dom}(h)$ and $v = h(\ell)$, and $H(\alpha) = 0$ otherwise, and analogously for a ghost heap and a prophecy heap.

We define the set of \emph{shared region invariants} $I \in \mathit{Invs} = \mathfrak{P}(\mathit{LogHeaps})$ as the predicates over logical heaps. We define the set of assertions $P \in \mathit{Asns} \subseteq \mathfrak{P}(\mathbb{N} \times \mathfrak{P}(\mathit{Invs}) \times \mathit{LogHeaps}) = \{P\ |\ \forall (n, A, H) \in P, n', A'.\;n' \le n \land A \subseteq A' \Rightarrow (n', A', H) \in P\}$ as the predicates over the length of the execution trace, the set of allocated shared region invariants, and the owned chunk fractions, that are closed under reducing the length of the execution trace and extending the set of allocated shared regions.

We define separating conjunction $P * P' = \{(n, A, H)\ |\ \exists H_1, H_2.\;H = H_1 + H_2 \land (n, A, H_1) \in P \land (n, A, H_2) \in P'\}$. $\circledast_{i\in I}\;P_i = P_{i_1} * \cdots * P_{i_n}$, where $I = \{i_1, \dots, i_n\}$ is a finite index set, denotes iterated separating conjunction. We define $\boxed{I} = \{(n, A, H)\ |\ I \in A\}$ and $\triangleright P = \{(n, A, H)\ |\ \forall n' < n.\;(n', A, H) \in P\}$. We define $\mathsf{True} = \{(n, A, H)\ |\ \mathsf{true}\}$.

As in \S\ref{sec:hoare-logic}, we define postconditions as functions from result values to assertions.

We define the \emph{ghost heaps} $g \in \mathit{GhostHeaps} = \mathit{GhostLocs} \rightharpoonup_\mathrm{fin} \mathit{GhostVals}$ as the finite partial functions from ghost locations to ghost values.

We define the precedence relation as follows, where $A, A' \subseteq \mathit{Invs}$ range over finite sets of allocated shared region invariants:
$$\begin{array}{l}
P \sqsubseteq P' = \forall (n, A, H) \in P, H_\mathrm{inv} \in \circledast_{I \in A}\;I, H_\mathrm{env}, V, h, \rho, g.\; H + H_\mathrm{inv} + H_\mathrm{env} \le V + h + \rho + g \Rightarrow\\
\quad \exists V', H', A' \supseteq A, H'_\mathrm{inv} \in \circledast_{I \in A'}\;I, g'.\;\mathsf{Traces}_N(V') \subseteq \mathsf{Traces}_N(V)\\
\quad\quad\land\; H' + H'_\mathrm{inv} + H_\mathrm{env} \le V' + h + \rho + g' \land (n, A', H') \in P'
\end{array}$$

We define the weakest precondition $\mathsf{wp}(e, Q) \in \mathit{Asns}$ of an expression with respect to a postcondition $Q$ as follows:
$$\begin{array}{l}
\mathsf{wp}(e, Q) = \{(n, A, H)\ |\ (\forall v.\;e = v \Rightarrow (n, A, H) \in Q(v)) \land \forall n' < n, T, V, h, g, \rho, H_\mathrm{inv}, H_\mathrm{env}.\\
\quad \mathsf{Traces}_N(V) \subseteq T \land H_\mathrm{inv} \in \circledast_{I \in A}\; I \land H + H_\mathrm{inv} + H_\mathrm{env} \le V + h + g + \rho \Rightarrow\\
\quad \lnot\mathsf{failed}(T, h, \rho, e) \land \forall T', h', \rho', e', \overline{e}_\mathrm{f}.\\
\quad\quad T, h, \rho, e \hookrightarrow T', h', \rho', e', \overline{e}_\mathrm{f} \Rightarrow \exists H', H'_\mathrm{inv}, V', g', A' \supseteq A.\\
\quad\quad\quad \mathsf{Traces}_N(V') \subseteq T' \land H'_\mathrm{inv} \in \circledast_{I \in A'}\; I \land H' + H'_\mathrm{inv} + H_\mathrm{env} \le V' + h' + \rho' + g'\\
\quad\quad\quad {\land}\; (n', A', H') \in \mathsf{wp}(e', Q) * \circledast_{e'' \in \overline{e}}\; \mathsf{wp}(e'', \mathsf{True})\}
\end{array}$$
We define $\{P\}\ e\ \{Q\} = (P \subseteq \mathsf{wp}(e, Q))$. We can derive both the Hoare rules from Fig.~\ref{fig:hoare-logic} and the new ones in Fig.~\ref{fig:hoare-conc}.
\begin{figure}
\begin{mathpar}
\inferrule{
\{P\}\ e\ \{\mathsf{True}\}
}{
\{P\}\ \mathbf{fork}(e)\ \{\mathsf{True}\}
}
\and
\inferrule{
\{I * P\}\ !\ell\ \{I * Q\}
}{
\{\boxed{I} * P\}\ \langle !\ell\rangle\ \{\boxed{I} * Q\}
}
\and
\inferrule{
v_\mathrm{old} \in \mathit{DataVals}\\
\{I * P\}\ \mathbf{if}\ !\ell = v_\mathrm{old}\ \mathbf{then}\ \ell \leftarrow v_\mathrm{new}; \mathbf{true}\ \mathbf{else}\ \mathbf{false}\ \{I * Q\}
}{
\{\boxed{I} * P\}\ \mathbf{CAS}(\ell, v_\mathrm{old}, v_\mathrm{new})\ \{\boxed{I} * Q\}
}
\and
I \sqsubseteq \boxed{I}
\and
\mathsf{True} \sqsubseteq \exists \hat{\ell}.\;\hat{\ell} \mapsto \hat{v}
\end{mathpar}
\caption{Proof rules of the concurrent Hoare logic}\label{fig:hoare-conc}
\end{figure}

We have adequacy:
\begin{theorem}
If $\{(n, A, H)\ |\ V \le H\} \subseteq \mathsf{wp}(e, \mathsf{True})$ then $\forall T, h, \rho.\;\mathsf{Traces}_N(V) \subseteq T \Rightarrow \mathsf{safe}((T, h, \rho, e))$.
\end{theorem}

\section{Verifying a channel implementation using prophecy variables}\label{sec:channels-proof}

Consider an implementation in Fig.~\ref{fig:channels-impl} of the channels specification introduced above. While this is not a very realistic implementation (for example, it performs busy waiting), it is sufficient to illustrate our proposed approach of using prophecy variables to achieve I/O-style specifications for in-memory data structures.
\begin{figure}
$$\begin{array}{l}
\mathbf{function}\ \mathsf{newChannel}()\ \{\\
\quad \mathbf{let}\ \mathsf{pvar} := \mathbf{create\_pvar}()\ \mathbf{in}\\
\quad \mathbf{let}\ \mathsf{queue} := \mathbf{ref}\ \epsilon\ \mathbf{in}\\
\quad \mathbf{return}\ \{\mathsf{pvar}:=\mathsf{pvar}; \mathsf{queue}:=\mathsf{queue}\}\\
\}\\
\\
\mathbf{function}\ \mathsf{send}(c, v)\ \{\\
\quad \mathbf{let}\ \mathsf{elems} := \langle!c.\mathsf{queue}\rangle\ \mathbf{in}\\
\quad \mathbf{if}\ \lnot\mathsf{CAS}(c.\mathsf{queue}, \mathsf{elems}, \mathsf{elems}\cdot v)\ \mathbf{then}\\
\quad\quad \mathsf{send}(c, v)\\
\}\\
\\
\mathbf{function}\ \mathsf{receive}(c)\ \{\\
\quad \mathbf{let}\ \mathsf{elems} := \langle!c.\mathsf{queue}\rangle\ \mathbf{in}\\
\quad \mathbf{match}\ \mathsf{elems}\ \mathbf{with}\\
\quad |\ \epsilon \Rightarrow \mathbf{return}\ \mathsf{receive}(c)\\
\quad |\ v\cdot\overline{v} \Rightarrow\\
\quad\quad \mathbf{if}\ \mathsf{CAS}(c.\mathsf{queue}, \mathsf{elems}, \overline{v})\ \mathbf{then}\ \{\\
\quad\quad\quad \mathbf{assign\_pvar}(c.\mathsf{pvar}, v); \mathbf{return}\ v\\
\quad\quad \}\ \mathbf{else}\\
\quad\quad\quad \mathbf{return}\ \mathsf{receive}(c)\\
\}
\end{array}$$
\caption{An implementation of the channels specifications. (We use ML-like structure notation as syntactic sugar for pair construction and destruction.)}\label{fig:channels-impl}
\end{figure}
The proof strategy is as follows. To allow the channel's queue to be accessed concurrently by multiple threads, we insert it into a shared region at channel creation time, whose invariant is as follows:
$$\begin{array}{l}
\mathsf{inv}(c, \hat{\ell}_\mathrm{S1}, \hat{\ell}_\mathrm{S2}, \hat{\ell}_\mathrm{R}) = \exists \overline{v}, \mu_1, \mu_2, M.\;c.\mathsf{queue} \mapsto \overline{v} * \hat{\ell}_\mathrm{S1} \stackrel{1/2}{\mapsto} \mu_1 * \hat{\ell}_\mathrm{S2} \stackrel{1/2}{\mapsto} \mu_2 * \hat{\ell}_\mathrm{R} \stackrel{1/2}{\mapsto} M\\
\quad {} \land \forall \mu' \in \mu_1\ ||\ \mu_2.\;\overline{v}\cdot\mu' \in M\\
\end{array}$$
It asserts full ownership of the queue, as well as \emph{fractional ownership} (with fraction one half) of three \emph{ghost cells} whose values track the state of the three threads using the channel. Furthermore, it asserts the consistency of the threads' states and the contents of the queue: these contents, followed by any interleaving of the sequences yet to be sent by the sender threads, satisfy the receiver thread's current prophecy variable constraint.

Predicates $\mathsf{sender}$ and $\mathsf{receiver}$ describe the resources and information held by the three threads:
$$\begin{array}{l}
\mathsf{sender}(c, \mu) = \exists \hat{\ell}_\mathrm{S}, \hat{\ell}_\mathrm{S1}, \hat{\ell}_\mathrm{S2}, \hat{\ell}_\mathrm{R}.\;\boxed{\mathsf{inv}(c, \hat{\ell}_\mathrm{S1}, \hat{\ell}_\mathrm{S2}, \hat{\ell}_\mathrm{R})} \land \hat{\ell}_\mathrm{S} \stackrel{1/2}{\mapsto} \mu \land \hat{\ell}_\mathrm{S} \in \{\hat{\ell}_\mathrm{S1}, \hat{\ell}_\mathrm{S2}\}\\
\mathsf{receiver}(c, \mu) = \exists \hat{\ell}_\mathrm{S1}, \hat{\ell}_\mathrm{S2}, \hat{\ell}_\mathrm{R}, M.\;\boxed{\mathsf{inv}(c, \hat{\ell}_\mathrm{S1}, \hat{\ell}_\mathrm{S2}, \hat{\ell}_\mathrm{R})}\\
\quad {} \land \mathsf{cpvar}(c.\mathsf{pvar}, M, \mu) * \hat{\ell}_\mathrm{R} \stackrel{1/2}{\mapsto} M
\end{array}$$
Notice that each thread owns a one-half fraction of the ghost cell tracking its state.

Owning just a fraction of a ghost cell does not allow a thread to mutate it; the following law allows a thread that performs an atomic operation to do so after matching up its fraction with the one held in the shared invariant:
$$\hat{\ell} \stackrel{1/2}{\mapsto} v_1 * \hat{\ell} \stackrel{1/2}{\mapsto} v_2 \Leftrightarrow \hat{\ell} \mapsto v_1 \land v_2 = v_1$$

We can now straightforwardly define the $\mathsf{send\_}$ and $\mathsf{receive\_}$ transition predicates:
$$\begin{array}{l}
\mathsf{send\_}(p, c, \mu) = p \sqsubseteq \mathsf{sender}(c, \mu)\\
\mathsf{receive\_}(p, c, \mu) = p \sqsubseteq \mathsf{receiver}(c, \mu)
\end{array}$$
Given these definitions, the proof of the channel implementation is straightforward.

We have developed a machine-checked proof that the chat server example and the channels implementation satisfy the specifications shown here by applying Iris to a programming language instrumented with I/O and prophecy variables as described in this paper \cite{iris-io-2-0}.

We have also encoded the approach into the logic of the VeriFast program verifier for C and verified C versions of the buffered I/O and chat server examples using VeriFast; see the \textsf{examples/abstract\_io} directory in the VeriFast 18.02 distribution \cite{verifast1802}. In this development, the constrained incremental prophecy variables are introduced as trusted primitives.

\section{Related work}\label{sec:related-work}

Penninckx \emph{et al.} \cite{io} originally proposed to use an embedding of Petri nets in separation logic to specify and verify I/O properties of programs in a way that is modular and compositional. The present paper builds on this work.

Besides \cite{io}, we are not aware of existing work that addresses the question of what specifications to use for platform I/O APIs, and for the program as a whole, such that a Hoare logic proof of the program's behavioral properties can be carried out modularly.

The idea of prophecy variables has been known for a long time; it was originally proposed by Abadi and Lamport \cite{abadi-lamport}. However, little work has appeared so far on formalizing its use in Hoare logics. Vafeiadis used it in his PhD thesis \cite{vafeiadis-phd}, but he did not formalize this aspect of his logic.

Zhang \emph{et al.} \cite{Zhang2012} have formalized a form of \emph{structural} prophecy variables. However, their form of prophecy variables does not appear to be suitable for implementing data structures with I/O-stype specifications. or example, it cannot be used to verify our chat server example.


\section{Conclusion}\label{sec:conclusion}

We propose an approach for assigning Hoare logic specifications to the I/O APIs of programming platforms, as well as to programs themselves, that allows the I/O behavior of these programs to be verified in a modular, compositional, and abstract manner. Compared to the existing work on which we build, we enable true I/O actions to be mixed in specifications transparently with actions that are implemented in-memory. We propose the use of constrained, incremental prophecy variables to allow nondeterministic in-memory operations to be specified like I/O input actions. Furthermore, we enable a greater degree of equivalence reasoning on specifications.

We have machine-checked the theory and the example proofs of the paper using the Iris library in the Coq proof assistant.


\bibliography{io2}

\appendix
\section{Appendix}

\subsection{I/O verification with Iris}

By instantiating the Iris program logic for our monitoring I/O semantics, we can use Iris' Hoare logic to verify programs that perform I/O.

In general, when using Iris to verify a program written in a programming language whose state space is $S$, ranged over by $s$, the assertions of the proof refer to the program state by directly or indirectly asserting \emph{fragmentary ownership} of a \emph{ghost cell} allocated at some well-known \emph{ghost cell address}, say $\gamma_\mathrm{STATE}$. For example, consider the Hoare triple below, which expresses that program $\mathit{foo}$ takes program state $s_0$ to program state $s_1$:

$$\annot{\{\dboxed{\circ s_0}^{\gamma_\mathrm{STATE}}\}}\ \mathit{foo}\ \annot{\{\dboxed{\circ s_1}^{\gamma_\mathrm{STATE}}\}}$$

The state space of our instrumented semantics consists of the states $\sigma$ of the original programming language (let's call these \emph{heaps}), and \emph{I/O states} $T$. For convenience, we will track these using separate ghost cells $\gamma_\mathrm{HEAP}$ and $\gamma_\mathrm{IO}$.

For example, assume $\mathsf{putchar} \in \mathcal{T}$ is a primitive I/O tag. Then $\mathit{hi} = \mathsf{putchar}(\texttt{'h'}); \mathsf{putchar}(\texttt{'i'})$ is a program, that satisfies specification $T = \{\tau\ |\ \tau \preceq \mathsf{putchar}(\texttt{'h'}, ()); \mathsf{putchar}(\texttt{'i'}, ())\}$, where $\preceq$ is prefixing: $\tau \preceq \tau \cdot \tau'$. We can express this in Iris as:

$$\annot{\{\dboxed{\circ T}^{\gamma_\mathrm{IO}}\}}\ \mathit{hi}\ \annot{\{\dboxed{\circ \{\epsilon\}}^{\gamma_\mathrm{IO}}\}}$$

This correctness judgment can be verified in Iris, using the following Hoare triple for primitive I/O commands:

$$\inferrule{
\exists v'.\;t(v, v') \in T
}{
\annot{\{\dboxed{\circ T}^{\gamma_\mathrm{IO}}\}}\ t(v)\ \annot{\{v'.\;\dboxed{\circ \{\tau\ |\ t(v, v')\cdot\tau \in T\}}^{\gamma_\mathrm{IO}}\}}
}$$

Note: we also have $\dboxed{\circ \emptyset}^{\gamma_\mathrm{IO}} \Rightarrow \mathsf{False}$.

(These specifications additionally express that the expressions do not access the heap.)

We can build the Petri nets-based logic from Sec.~\ref{sec:petrinets} on top of this basic logic, inside Iris, as follows. We introduce a ghost cell $\gamma_\mathrm{PETRI}$ that tracks the Petri net and the marking, and an invariant that links it to $\gamma_\mathrm{IO}$:
$$\begin{array}{l}
\dboxed{\circ \mathsf{traces}_N(V)}^{\gamma_\mathrm{IO}} \Rrightarrow\\
\quad \boxed{\exists V, T.\;\dboxed{\circ T}^{\gamma_\mathrm{IO}} * \dboxed{\bullet (N, V)}^{\gamma_\mathrm{PETRI}} \land \mathsf{traces}_N(V) \subseteq T} * \dboxed{\circ (N, V)}^{\gamma_\mathrm{PETRI}}
\end{array}$$

Furthermore, we allow splitting of markings: $$\dboxed{\circ (N, V_1 \uplus V_2)}^{\gamma_\mathrm{PETRI}} \Leftrightarrow \dboxed{\circ (N, V_1)}^{\gamma_\mathrm{PETRI}} * \dboxed{\circ (N, V_2)}^{\gamma_\mathrm{PETRI}}$$

This way, we can define the primitive predicates of Sec.~\ref{sec:petrinets} as follows:
$$\mathsf{token}(p) = \dboxed{\circ (\_, \llbrace p\rrbrace)}^{\gamma_\mathrm{PETRI}}$$
$$t\mathsf{\_}(p, v, v', q) = \exists N.\;\dboxed{\circ (N, \mathbf{0})}^{\gamma_\mathrm{PETRI}} \land t(p, v, v', q) \in N$$
We now wish to establish the following proof rule for primitive I/O commands:
$$\annot{\{\mathsf{token}(p) \land t{\mathsf{\_}}(p, v, v', q)\}}\ t(v)\ \annot{\{v''.\;v'' = v' \land \mathsf{token}(q)\}}$$

However, this proof rule is sound only if the Petri net is \emph{result-deterministic}, by which we mean that if $\{\tau \cdot t(v, v'), \tau \cdot t(v, v'')\} \subseteq \mathsf{traces}_N(V)$, then $v'' = v'$. We will therefore assume the Petri net is result-deterministic. (Furthermore, we extend the invariant above with a conjunct saying that $T$ is result-deterministic.) However, using a single result-deterministic Petri net as a program specification precludes underspecifying the environment. We will therefore use as a program specification not a single Petri net, but a set of result-deterministic Petri nets. More specifically, we will verify the program not against a single Petri net, but against all Petri nets from some set. We are in fact applying the following law, where $\mathbb{T}$ is a set of sets of traces:

$$(\forall T \in \mathbb{T}.\;e \vDash T) \Rightarrow e \vDash \bigcup \mathbb{T}$$

For example, in the $\mathsf{toUpper}$ example from Sec.~\ref{sec:petrinets} the program is verified against 26 distinct Petri nets, each of which determines the result of the $\mathsf{getchar}$ command differently.

\end{document}